\documentclass[aps,prx,superscriptaddress,twocolumn,longbibliography,floatfix,reprint]{revtex4-2}
\usepackage{amsmath,amssymb,amsthm,mathtools,bm}
\usepackage{aliascnt}
\usepackage{booktabs,multirow}
\usepackage{enumitem}
\usepackage[dvipsnames]{xcolor}
\colorlet{BLUE}{blue}
\usepackage{tikz}
\usetikzlibrary{arrows.meta,calc,positioning,patterns,decorations.pathreplacing}
\usepackage[most]{tcolorbox}
\usepackage[colorlinks=true,linkcolor=blue, citecolor=blue]{hyperref}
\usepackage[nameinlink,capitalize]{cleveref}
\allowdisplaybreaks
\usepackage{appendix}
\usepackage[normalem]{ulem}

\usepackage{tikz}
\usetikzlibrary{quantikz2}

\newcommand{\invgate}[1]{ \gate[1,style={blue!0,fill=blue!0,fill opacity=0,opacity=0}]{#1} }

\definecolor{natureblue}{HTML}{096B72}
\definecolor{natureorange}{HTML}{D87922}
\definecolor{natureink}{HTML}{1E2933}
\definecolor{naturegrey}{HTML}{E9EEF0}
\definecolor{naturepale}{HTML}{F5F8F8}

\newcommand{\id}{\mathrm{id}}
\newcommand{\cN}{\mathcal N}
\newcommand{\cM}{\mathcal M}

\newcommand{\Sep}{\mathrm{Sep}}
\newcommand{\supp}{\mathrm{supp}}
\newcommand{\wt}{\mathrm{wt}}
\newcommand{\Tr}{\operatorname{Tr}}
\newcommand{\Ov}{\operatorname{Ov}}
\newcommand{\Echi}{E_{\chi}}
\newcommand{\Venv}{C_{\mathrm{lay}}}
\newcommand{\Vpath}{C_{\mathrm{path}}}
\newcommand{\Vav}{C_{\mathrm{av}}}
\newcommand{\Videal}{C_{\mathrm{ideal}}}
\newcommand{\CFT}{C_{\mathrm{FT}}}
\newcommand{\Cmap}{C_{\mathrm{st}}}
\newcommand{\cgad}{c_{\mathrm{gad}}}
\newcommand{\Nst}{N_{\mathrm{st}}}
\newcommand{\dst}{d_{\mathrm{st}}}
\newcommand{\Qpeak}{Q_{\mathrm{peak}}}
\newcommand{\PL}{P_{\mathrm L}}
\newcommand{\qmin}{q_{\min}}
\newcommand{\qmax}{q_{\max}}
\newcommand{\cloc}{c_{\mathrm{loc}}}
\newcommand{\eps}{\varepsilon}

\newcommand{\cG}{\mathcal{G}}
\newcommand{\cS}{\mathcal{S}}

\renewcommand{\P}{\mathsf{P}}

\renewcommand{\>}{{\rangle}}
\newcommand{\<}{{\langle}}
\newcommand{\ketbra}[1]{|{#1}\>\mkern-4mu\<{#1}|}

\newtheorem{theorem}{Theorem}
\newaliascnt{lemma}{theorem}
\newtheorem{lemma}[lemma]{Lemma}
\aliascntresetthe{lemma}
\newaliascnt{corollary}{theorem}
\newtheorem{corollary}[corollary]{Corollary}
\aliascntresetthe{corollary}
\newaliascnt{proposition}{theorem}
\newtheorem{proposition}[proposition]{Proposition}
\aliascntresetthe{proposition}
\theoremstyle{definition}
\newaliascnt{definition}{theorem}
\newtheorem{definition}[definition]{Definition}
\aliascntresetthe{definition}
\newaliascnt{remark}{theorem}
\newtheorem{remark}[remark]{Remark}
\aliascntresetthe{remark}
\newaliascnt{assumption}{theorem}
\newtheorem{assumption}[assumption]{Assumption}
\aliascntresetthe{assumption}

\newtcolorbox{resultbox}[1][]{enhanced,breakable,colback=naturepale,colframe=natureblue,
  boxrule=0.75pt,arc=1.5pt,left=7pt,right=7pt,top=6pt,bottom=6pt,#1}
\newtcolorbox{scopebox}[1][]{enhanced,breakable,colback=white,colframe=natureorange,
  boxrule=0.65pt,arc=1.5pt,left=7pt,right=7pt,top=5pt,bottom=5pt,#1}

\newcommand{\CQT}{Centre for Quantum Technologies, National University of Singapore, 3 Science Drive 2, Singapore 117543\looseness=-1}
\newcommand{\NTU}{Nanyang Quantum Hub, School of Physical and Mathematical Sciences, Nanyang Technological University, Singapore 639673\looseness=-1}
\def\QUICS{QuICS, NIST/University of Maryland, College Park, Maryland 20742, USA}
\def\UMIACS{UMIACS, University of Maryland, College Park, Maryland 20742, USA}
\def\TII{Quantum Research Center, Technology Innovation Institute, Abu Dhabi, United Arab Emirates}

\definecolor{darkred}{RGB}{110,0,0}
\newcommand{\atchg}[1]{\textcolor{darkred}{#1}}

\definecolor{darkblue}{RGB}{0,0,190}

\newcommand{\SMLong}{Supplementary Information}

\newcommand{\SM}{SI}

\begin{document}

\title{Fault-tolerant quantum computation cannot be achieved with \\ constant spacetime overhead}

\author{Kishor Bharti}
\affiliation{\QUICS}
\affiliation{\UMIACS}

\author{Tobias Haug}
\affiliation{\TII}

\author{Andrew Tanggara}
\affiliation{\NTU}
\affiliation{\CQT}

\begin{abstract}

The threshold theorem states that quantum computations can be made reliable below a physical error threshold, at the cost of additional physical qubits and circuit depth. 
Recent work has reduced these space and time overheads to polylogarithmic or nearly logarithmic scalings, but whether the \emph{cumulative} spacetime overhead can be constant has remained unclear.
Here, we show that even for the simplest task of preserving quantum information in a quantum memory, under an optimistic noise model and allowing general adaptive protocols, there is an unavoidable logarithmic contribution to the cumulative spacetime overhead.
This additional cost can nevertheless be shared among many logical qubits, so sufficiently wide computations, including standard implementations of Shor's algorithm, may still achieve constant relative overhead.
We further give a positive-rate CSS code construction that attains the memory bound, identify sufficient conditions under which the same scaling extends from quantum memory to fault-tolerant circuit implementations, and derive circuit-size bounds for subsystem spacetime codes.
Our work establishes fundamental limits on the resources required for quantum fault tolerance.

\end{abstract}

\maketitle

 \let\oldaddcontentsline\addcontentsline%
\renewcommand{\addcontentsline}[3]{}%

The threshold theorem establishes that arbitrarily long quantum computations can be made reliable below a physical error threshold, with explicit upper bounds on the required overhead~\cite{Shor1996,AharonovBenOr1997,Kitaev1997,Knill1998,AGP2006,Preskill1998}. What it does not determine is how small this overhead can ultimately be. In particular, what is the minimum cumulative physical cost required to sustain a computation of a desired duration and target accuracy? A central question is whether fault tolerance can be achieved with constant \emph{cumulative} overhead. Space overhead measures the number of physical qubits required at a given time, whereas cumulative spacetime overhead accounts for both how many physical qubits are used and for how long they are required.

Previous work has made substantial progress on the former: Gottesman's constant-overhead framework and subsequent quantum low-density parity-check (qLDPC) code constructions allow the number of physical qubits to scale linearly with the logical width~\cite{Gottesman2014,FGL2018,PanteleevKalachev2022,LeverrierZemor2022}, even for general circuit noise~\cite{ChristandlFawziGoswami2025}.
Concrete qLDPC architectures based on long-range connectivity
and reconfigurable atom arrays achieve low overhead in their respective
settings~\cite{CohenEtAl2022,XuAtomArrays2024}. Recent protocols combine constant space overhead with reduced time overhead~\cite{YamasakiKoashi2024,Tamiya2026,NguyenPattison2025}, while related approaches achieve low-overhead logical operations for high-rate codes~\cite{WilliamsonYoder2026,Cowtan2025}. 

Lower-bound results point to fundamental limitations on these resources. Fawzi, M\"uller-Hermes, and Shayeghi showed that, under non-unitary noise, the required physical width cannot remain independent of the computation duration~\cite{FMHS2022}, with related converse and geometric bounds obtained in Refs.~\cite{Uthirakalyani2023,Baspin2023}. These results constrain instantaneous resources, but do not determine the minimum cumulative circuit size or its dependence on the target error. Whether constant cumulative overhead is possible, particularly for adaptive fault-tolerant protocols, has therefore remained open.

{
Here, we address this question by considering quantum memory, arguably the simplest fault-tolerant task: preserving an arbitrary quantum state for a prescribed time. Despite this simplicity, memory already captures the fundamental cost of maintaining quantum information in the presence of noise. For independent erasure noise, even allowing arbitrary adaptive protocols and ideal recovery, we determine this cost exactly up to constant factors. Storing $K$ logical qubits for $S$ time steps with target error $\varepsilon$ requires
\begin{equation}
    C_{\min}(K,S,\varepsilon)
    =\Theta\!\left(
    S\left(K+\log\!\left(\frac{S}{\varepsilon}\right)\right)
    \right).
\end{equation}
Here, $C_{\min}$ denotes the minimum worst-case number of physical storage locations, where one storage location corresponds to keeping one physical qubit for one time step. Such locations are the points in spacetime at which stored qubits remain exposed to noise; the formal adaptive model and error criterion are given in \hyperref[sec:methods-memory]{Methods~\ref*{sec:methods-memory}} and \hyperref[sec:supp-memory-model]{\SMLong{} (\SM{})~\ref*{sec:supp-memory-model}}.

The two contributions have a simple interpretation. The term $SK$ is the cost of retaining the logical information itself, whereas $S\log(S/\varepsilon)$ is an additional reliability cost required to keep the probability of losing the stored information sufficiently small over the full duration. This second contribution is unavoidable even when the protocol adapts its physical resources to previous measurement outcomes. Conversely, a positive-rate CSS-code construction with ideal recovery achieves the same scaling, making the bound tight.

The result reveals a width--duration tradeoff in fault tolerance. The minimum relative overhead scales as
\begin{equation}
    \Theta\!\left(
    1+\frac{\log(S/\varepsilon)}{K}
    \right).
\end{equation}
For a fixed or small logical register, the overhead therefore grows with the storage duration; in particular, for constant or inverse-polynomial target error and $K=O(1)$, it grows logarithmically. By contrast, when $K=\Omega(\log(S/\varepsilon))$, the additional reliability cost can be shared among many logical qubits and the relative overhead can remain constant. Thus, constant overhead can be possible for sufficiently wide computations, but no single constant can bound the spacetime overhead uniformly over all widths and prescribed durations. \cref{fig:memory-crossover} summarizes this crossover.

Beyond quantum memory, we show that the same qualitative limitation persists for general non-unitary qubit noise, and we identify sufficient conditions under which the same width--reliability tradeoff extends to %
fault-tolerant circuits.  %
Under these conditions, wide polynomial-depth computations, including standard implementations of Shor's algorithm, can amortize the additional reliability cost, whereas Grover search lies near the crossover and protocols that repeatedly reuse a small coherent register, such as iterative phase estimation or long-time simulation, can enter the regime in which this cost dominates. Finally, we derive complementary bounds for subsystem spacetime codes, linking circuit resources to code distance and logical error. Taken together, these results show that fault-tolerance overhead is governed not only by how much quantum information is processed, but also by how long that information must remain protected.
}

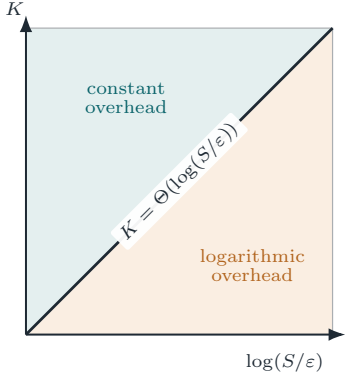
\begin{figure}[t]
\centering
\begin{tikzpicture}[x=0.78cm,y=0.78cm,>=Latex,font=\scriptsize]
  \def\L{5.20}
  \path[fill=natureorange!13] (0,0) rectangle (\L,\L);
  \path[fill=natureblue!11] (0,0)--(0,\L)--(\L,\L)--cycle;
  \draw[natureink!38,line width=0.45pt] (0,0) rectangle (\L,\L);
  \draw[-Latex,line width=0.75pt,natureink] (0,0)--(\L+0.22,0);
  \draw[-Latex,line width=0.75pt,natureink] (0,0)--(0,\L+0.22);
  \draw[natureink,line width=1.05pt] (0,0)--(\L,\L);

  \node[anchor=south west,natureink] at (-0.50,\L+ 0.08) {$K$};
  \node[anchor=north east,natureink] at (\L,-0.16)
    {$\log(S/\eps)$};

  \node[align=center,text=natureblue!92!black] at (1.70,4.02)
    {constant\\[-1pt]overhead};
  \node[align=center,text=natureorange!82!black] at (3.84,1.18)
    {logarithmic\\[-1pt]overhead};

  \node[rotate=45,fill=white,fill opacity=0.95,text opacity=1,
    rounded corners=1pt,inner xsep=2.4pt,inner ysep=1.3pt,text=natureink]
    at (2.60,2.60)
    {$K=\Theta\!\left(\log(S/\eps)\right)$};
\end{tikzpicture}
\caption{
The optimal relative overhead under independent erasure noise is controlled by the ratio $\log(S/\epsilon)/K$, where $K$ is the number of logical qubits, $S$ the storage duration and $\epsilon$ the target error. At the crossover $K=\Theta(\log(S/\epsilon))$, the costs of retaining the logical information and maintaining reliability are comparable. For $K=\Omega(\log(S/\epsilon))$, the relative overhead remains constant, whereas for $K=o(\log(S/\epsilon))$ the reliability cost dominates and the overhead grows logarithmically. %
}
\label{fig:memory-crossover}
\end{figure}

\section{Circuit-size bounds for erasure noise}
\label{sec:main-erasure}
{
Let us consider an adaptive physical circuit $\mathcal Q$ that stores $K$ logical qubits for $S$ time steps. We call it a $(K,S,\varepsilon)$ memory protocol if, after decoding, its action differs from perfect storage, the identity channel on the $K$ logical qubits, by at most $\varepsilon$ in diamond norm.
To quantify the cumulative cost of the protocol, we count the number of physical circuit locations, where each gate, measurement, preparation or period of storage represents a point at which noise can act.
Because the protocol may adapt to measurement outcomes, different executions can use different numbers of physical storage locations. For a complete measurement record $h_S$ with nonzero probability, let $C(\mathcal Q;h_S)$ denote the number of storage locations used in that execution. We define $\Vpath(\mathcal Q)$ as the largest such cost over all possible measurement records, and $C_{\min}(K,S,\varepsilon)$ as the smallest worst-case cost among all protocols achieving the target accuracy.}

First, we consider an erasure noise model, where at every time step, each physical qubit is independently erased with probability $p$, and the erased locations are known to the recovery operation. If all live physical qubits are erased at one time step, the stored quantum state is lost and no later recovery can reconstruct an arbitrary input.
{This model is not only analytically convenient: erasure-qubit
architectures aim to convert dominant physical faults into erasures at
known locations, providing a concrete setting in which located-erasure
noise is physically relevant~\cite{GuRetzkerKubica2025}.}

To derive the reliability contribution, first consider a one-qubit adaptive memory protocol $\mathcal Q$ and write $C=C_{\mathrm{path}}(\mathcal Q)$. Let $P_{\mathrm{surv}}(S,C)$ denote the probability that no time step erases the entire live register. Even when the physical width is chosen adaptively,
\begin{equation}
P_{\mathrm{surv}}(S,C)\leq \exp\!\left(-S p^{C/S}\right).
\label{eq:main-survival}
\end{equation}
We reduce the adaptive optimization to a deterministic allocation of the available storage locations and then apply Jensen's inequality. To connect this event to logical error, {we apply the memory} to one half of a Bell pair and let $F_{\mathrm e}$ be the final entanglement fidelity. On a complete-erasure event, the output is separable from the reference and has Bell-state fidelity at most $1/2$. Thus $F_{\mathrm e}\leq P_{\mathrm{surv}}+(1-P_{\mathrm{surv}})/2$, and hence
\begin{equation}
1-F_{\mathrm e}\geq \frac12\left(1-e^{-S p^{C/S}}\right).
\label{eq:main-erasure-fidelity}
\end{equation}
If the diamond error is at most $\eps<1$, define $c_{\eps}=-\log(1-\eps)$. For $S>c_{\eps}$,
\begin{equation}
C\geq \frac{S}{\log(1/p)}\log\!\frac{S}{c_{\eps}}.
\label{eq:main-erasure-one}
\end{equation}
We obtain $\Vpath\geq SK$ {from a separate Schmidt-number argument}. {Combining these inequalities, we obtain} the lower bound in \cref{thm:tight-erasure}.
We state the bounds and their proof idea here, derive them in \hyperref[sec:methods-erasure-lower]{Methods~\ref*{sec:methods-erasure-lower}}, and give the complete adaptive proofs in \hyperref[sec:supp-erasure-lower]{\SM{}~\ref*{sec:supp-erasure-lower}}.

Let $\delta_{\mathrm{GV}}\simeq0.1100$  (see \hyperref[sec:methods-css-upper]{Methods~\ref*{sec:methods-css-upper}} and \hyperref[sec:supp-css-upper]{\SM{}~\ref*{sec:supp-css-upper}} for details). {For $p<\delta_{\mathrm{GV}}$, we use the CSS Gilbert--Varshamov bound to choose} positive-rate CSS codes whose distance is linear in the block length and whose relative distance can be chosen larger than $p$~\cite{CalderbankShor1996,Ashikhmin2014}. We then use a Chernoff bound to show that the probability of an uncorrectable erasure pattern decreases exponentially in the block length. We choose
\begin{equation}
N=O\!\left(K+\log (S/\varepsilon)\right),
\label{eq:main-block}
\end{equation}
which both encodes the $K$ logical qubits and suppresses the total error over $S$ time steps. Applying ideal recovery after each time step and the telescoping inequality, we bound the total diamond-norm error by $\eps$. 
We give the binary-entropy definition of $\delta_{\mathrm{GV}}$ and the parameter choices in \hyperref[sec:methods-css-upper]{Methods~\ref*{sec:methods-css-upper}}, and we give the complete construction and tightness statement in~\hyperref[sec:supp-css-upper]{\SM{}~\ref*{sec:supp-css-upper}}.

\begin{theorem}
\label{thm:tight-erasure}
Let $C_{\min}(K,S,\eps)$ denote the infimum of the worst-case circuit size over adaptive $K$-qubit memories with ideal recovery, where the worst case is over realizable complete measurement records with nonzero probability.
Fix $0<p<\delta_{\mathrm{GV}}$ and $0<\eps_0<1$. For independent erasure noise  and total diamond-norm error of $\eps\in(0,\eps_0]$,
\begin{equation}
C_{\min}(K,S,\eps)=
\Theta\!\left(S\left(K+\log (S/\varepsilon)\right)\right).
\label{eq:main-tight}
\end{equation}
We state the parameter dependence precisely in
\hyperref[cor:tight-erasure]{Corollary~\ref*{cor:tight-erasure}}.
\end{theorem}

Dividing by the ideal circuit size $KS$ gives
\begin{equation}
\frac{C_{\min}(K,S,\eps)}{KS}=\Theta\!\left(1+\frac{\log(S/\eps)}{K}\right).
\label{eq:main-relative}
\end{equation}
The reliability term dominates when $K=o(\log(S/\eps))$, is comparable to the dimension term when $K=\Theta(\log(S/\eps))$, and is lower order when $K=\omega(\log(S/\eps))$. Thus bounded relative overhead already holds for $K=\Omega(\log(S/\eps))$. The positive-rate code does not eliminate this term; it distributes the additional locations among the logical qubits in the block.

For example, let $p=10^{-2}$, $S=10^9$, and $\eps=10^{-3}$. Then $c_\eps=-\log(1-\eps)\simeq1.0005\times10^{-3}$, and the one-qubit bound in \cref{eq:main-erasure-one} gives
\begin{equation}
\frac{C}{S}\geq
\frac{\log(S/c_\eps)}{\log(1/p)}\simeq6.00.
\end{equation}
Moreover, $\log(S/\eps)\simeq27.63$, so $1+\log(S/\eps)/K$ equals $28.63$, $1.86$, $1.22$, and $1.027$ for $K=1,32,128$, and $1024$, respectively.
These values illustrate the crossover parameter in \cref{eq:main-relative}; finite-device estimates additionally depend on the constants and architecture-specific location counts.

\section{Circuit-level upper bound}
\label{sec:main-circuit-upper}

A memory code does not by itself give a fault-tolerant implementation of a general circuit. We therefore state conditions under which the memory scaling extends to logical gates. We assume a positive-rate code family with available block lengths within a constant factor of any sufficiently large target length. A length-$N$ code block consists of $N$ physical qubits; positive rate means that it encodes a number of logical qubits proportional to $N$. 
An entire logical layer on a block of length $N$ must have a gadget using at most $c_qN$ physical qubits for at most $c_t$ time steps. For every preceding measurement record, the implemented joint quantum--classical channel, tensored with the identity on an arbitrary external system that may be entangled with the logical input and on which neither channel acts, must be within diamond-norm distance $Ae^{-\beta N}$ of the corresponding ideal logical layer. 
Here $\beta>0$ is the block-length error-suppression exponent and $A\geq1$ is a fixed prefactor.
We include all quantum systems retained between gadgets in the circuit-size count and state the formal assumptions in \hyperref[sec:methods-circuit-upper]{Methods~\ref*{sec:methods-circuit-upper}}. 
Here $T$ is the ideal circuit depth and $S$ is the implementation depth, so $S\leq c_tT$. We write $\CFT$ for the worst-case number of physical circuit locations over complete %
measurement records.

\begin{theorem}
\label{thm:conditional-circuit}
    \textit{Conditional circuit upper bound.} Under these assumptions, an ideal circuit of width $K$, depth $T$, and target diamond-norm error $\eps$ admits a fault-tolerant implementation whose worst-case physical circuit size satisfies
    \begin{equation}
    \CFT
    \leq \cgad T\!\left(K+1+\log\frac{AT}{\eps}\right),
    \label{eq:main-circuit}
    \end{equation}
    where $\cgad>0$ depends only on the code-and-gadget family.
\end{theorem}

We choose a block length $N=O(K+\log(T/\eps))$. The rate condition ensures that the block encodes all $K$ logical qubits, while the telescoping inequality bounds the total error by $TAe^{-\beta N}$. Since each logical layer uses $O(N)$ physical qubits for $O(1)$ time steps, the total circuit size is $O(TN)$. {After dividing by the ideal circuit size $KT$, we obtain} relative overhead $O\!\left(1+\log(T/\eps)/K\right)$. {We give the derivation in} \hyperref[sec:methods-circuit-upper]{Methods~\ref*{sec:methods-circuit-upper}}.

Existing general-purpose compilers with constant space overhead realize a different point in the width--depth tradeoff by allowing nonconstant time overhead~\cite{YamasakiKoashi2024,Tamiya2026,NguyenPattison2025}. Under the stated assumptions, relative overhead is bounded when $K=\Omega(\log(T/\eps))$, and the logarithmic term is lower order when $K=\omega(\log(T/\eps))$; we give examples in \hyperref[sec:algorithmic-regimes]{\SM{}~\ref*{sec:algorithmic-regimes}}.

\section{Subsystem spacetime-code bounds}
\label{sec:main-spacetime}

{\cref{thm:tight-erasure} applies directly to general adaptive memory protocols, without requiring a circuit-to-code representation. We next consider fault-tolerant circuits that admit an exact spacetime-code description, which allows circuit resources to be related to code distance and logical error. Motivated by circuit-to-code, qLDPC, and fault-complex representations of fault-tolerant circuits~\cite{Bacon2017,
DelfossePaetznick2023,Li2025LDPC,HillmannEtAl2025,Pesah2025}, as well as recent low-overhead constructions based on spacetime lifting~\cite{XuWangLiu2026}, we consider circuits for which selected Pauli fault locations define an exact binary linear subsystem spacetime code.}

Let its parameters be $[[\Nst,k,r,\dst]]$, where $k$ and $r$ are the numbers of protected and gauge qubits, $\Nst$ is the number of selected fault coordinates, and $\dst$ is the dressed distance. The subsystem Singleton bound gives
\begin{equation}
\Nst\geq k+r+2(\dst-1).
\label{eq:main-singleton}
\end{equation}
Let $\Cmap$ %
denote the physical circuit-location count associated with the circuit-to-code map, which takes a chosen fault-tolerant Clifford circuit and selected elementary Pauli fault locations to its exact binary linear subsystem spacetime code.
This quantity is not automatically $\Vpath$: the latter counts storage locations in the memory model, whereas $\Cmap$ depends on the locations included in the explicit spacetime-code construction. 
{Assume that $\Nst\leq\cloc \Cmap$ for a scale-independent constant $\cloc$.}
This assumption says that the construction assigns at most a constant number of spacetime-code coordinates per counted circuit location. Without it, a bound on the code length need not give a bound on the physical circuit size.
{Combining this assumption with \cref{eq:main-singleton}, we obtain}
\begin{equation}
\Cmap\geq \frac{k+r+2(\dst-1)}{\cloc}.
\label{eq:main-volume}
\end{equation}

Let $\PL$ be the probability that syndrome-based Pauli recovery leaves a nontrivial dressed logical Pauli. For a dressed logical Pauli $L$, define $\Ov_{\nu}(L)$ as the overlap between the fault distribution $\nu(F)$ and its translate $\nu(FL)$. Since $F$ and $FL$ have the same syndrome but cannot both be corrected,
\begin{equation}
\PL\geq\frac12\Ov_{\nu}(L),\qquad
\Ov_{\nu}(L)=\sum_F\min\{\nu(F),\nu(FL)\}.
\label{eq:main-pairing}
\end{equation}
For full-support product Pauli noise, let $\qmin$ and $\qmax$ be the smallest and largest single-coordinate probabilities of Pauli noise operators $I,X,Y,Z$, and set $\rho=\qmin/\qmax$. {Choosing a minimum-weight $L$, we obtain}
\begin{equation}
\PL\geq\frac{\rho^{\dst}}{1+\rho^{\dst}}.
\label{eq:main-pair-product}
\end{equation}
Consequently, if $\qmin<\qmax$ and $\PL\leq\eps<1/2$, define
\begin{equation}
d_\eps^\star=
\frac{\log((1-\eps)/\eps)}{\log(\qmax/\qmin)}.
\end{equation}
Thus $d_\eps^\star$ is the minimum spacetime distance required by the target logical error $\eps$ under the stated product\atchg{-Pauli}-noise model.
Then
\begin{equation}
\begin{aligned}
\dst&\geq d_\eps^\star,\\
\Nst&\geq\max\!\left\{k+r,\,
k+r+2(d_\eps^\star-1)\right\}.
\end{aligned}
\label{eq:main-distance}
\end{equation}
{We therefore find that vanishing logical Pauli error requires} growing distance and growing absolute redundancy under the stated assumptions. We give the derivation in \hyperref[sec:methods-spacetime]{Methods~\ref*{sec:methods-spacetime}}, and the complete assumptions and proofs in \hyperref[sec:supp-spacetime]{\SM{}~\ref*{sec:supp-spacetime}}.

\section{Discussion}

{
Our results show that fault tolerance carries an unavoidable cumulative cost that is not captured by space overhead alone. Even for the elementary task of preserving quantum information, maintaining reliability for longer times requires additional physical resources. For a fixed number of logical qubits, this leads to a logarithmically growing relative spacetime overhead. Thus, although fault-tolerant architectures can achieve constant space overhead, no single constant can bound the cumulative overhead across all logical widths and durations.

This limitation has an important qualification. The additional reliability cost can be shared across many logical qubits. Consequently, sufficiently wide computations can retain constant relative overhead, whereas protocols that repeatedly reuse a small coherent register are fundamentally more constrained. Fault-tolerance overhead is therefore governed not only by the overall size of a computation, but by how its quantum information is distributed across space and time. This distinction separates wide polynomial-depth computations, such as standard implementations of Shor's algorithm, from long-lived memories, iterative phase estimation and other protocols that keep a small quantum register coherent for many rounds.

Our tight result is obtained for independent erasure noise and is attained by a positive-rate CSS-code construction with ideal recovery. The broader results indicate that the same physical mechanism persists beyond this setting: non-unitary noise imposes a growing reliability cost (see~\hyperref[sec:supp-general-noise]{\SM{}~\ref*{sec:supp-general-noise}}), while suitable code and gadget properties can allow this cost to be amortized in fault-tolerant circuits. These results complement recent progress towards constant-space and low-overhead fault tolerance~\cite{Gottesman2014,FGL2018,YamasakiKoashi2024,Tamiya2026,NguyenPattison2025} by identifying when constant cumulative overhead is, and is not, possible.

An important next step is to understand how these fundamental bounds are modified by the constraints of realistic architectures, including locality, decoding, communication and non-Clifford resource preparation. More generally, the width--duration trade-off identified here should be relevant whenever a small quantum system must remain coherent for a long time, including quantum memories, networking, sensing and long-time simulation. It provides a simple principle for fault-tolerant design: reliability costs can be amortized over space, but they cannot be eliminated.
}

\section*{Methods}

\setcounter{subsection}{0}
\renewcommand{\thesubsection}{\Alph{subsection}}

\makeatletter
\renewcommand{\p@subsection}{}
\makeatother

\renewcommand{\theHsubsection}{methods.\Alph{subsection}}

\subsection{Memory model and error criterion}
\label{sec:methods-memory}
We first define the adaptive memory protocol and the circuit-size quantities used in the main theorem.
The adaptive protocol $\mathcal Q$ evolves for $S$ time steps, with the $s$th application of the physical noise channel denoted by $\partial_s$. Immediately before $\partial_s$, the protocol uses only the preceding classical measurement record $h_{s-1}$ to select its complete physical quantum register
\begin{equation}
\begin{aligned}
\mathcal H_s^{\mathrm{live}}(h_{s-1})
  &=\bigotimes_{a\in\Gamma_s(h_{s-1})}\mathbb C_a^2,\\
q_s(h_{s-1})
  &=|\Gamma_s(h_{s-1})|
\end{aligned}
\label{eq:methods-live}
\end{equation}
where $\Gamma_s(h_{s-1})$ is the complete physical register. 
Every qubit in this tensor product is acted on by the specified noise channel and contributes one storage location, namely the standard wait location for one qubit during one time step. After the noise application, an ideal controller may measure, discard, reset, or introduce quantum systems and update the classical record. Classical records are not included in the quantum circuit size. Only systems present at a modeled time step enter this storage-location count; a circuit-level model additionally counts the preparations, gates, measurements, and waits used by the intervening control.

This model is the storage-location restriction of the standard circuit-location model used in threshold proofs. Refining a modeled time step into its constituent preparation, gate, measurement, recovery, and wait locations can only increase the location count. Consequently, a lower bound on the storage-location count is also a lower bound on the total number of locations in a gate-level fault-tolerant circuit.

The peak width, the circuit size of an execution, the {worst-case} circuit size, and the mean circuit size are, respectively,
\begin{equation}
\begin{aligned}
\Qpeak&=\max_s\sup_{h_{s-1}}q_s,\\
{C(\mathcal Q;h_S)}&=\sum_s q_s(h_{s-1}),\\
\Vpath(\mathcal Q)&=\sup_{h_S}C(\mathcal Q;h_S),\\
\Vav(\mathcal Q)&=\mathbb E_{h_S}C(\mathcal Q;h_S).
\end{aligned}
\end{equation}
These quantities satisfy $\Vav(\mathcal Q)\leq\Vpath(\mathcal Q)\leq\Qpeak S$. We define $C_{\min}(K,S,\eps)$ as {$\inf_{\mathcal Q}\Vpath(\mathcal Q)$ over adaptive protocols whose decoded memory channel has diamond-norm error at most $\eps$}. {When the protocol is clear, we suppress the argument $\mathcal Q$.} 
The fixed-duration requirement states that the unknown quantum information remains in the counted physical system until the prescribed final time. In the memory theorem, {we compare the ideal and encoded memories} over the same $S$ time steps, so $C_{\mathrm{ideal}}=KS$. In the circuit-level theorem, $T$ denotes the ideal circuit depth and $S$ denotes the implementation depth.

For a decoded memory channel $\cM$, the error is $\|\cM-\id\|_\diamond$, with the reference system included in the definition of the diamond norm. In the spacetime-code results, $\PL$ instead denotes the probability that Pauli recovery leaves a nontrivial dressed logical Pauli. {We use the two error measures} in separate statements.

\subsection{Time-step convention}
\label{sec:methods-time-step}

We use the standard synchronized circuit model, in which each discrete time step has duration $\tau>0$.
When elementary operations have different durations, {we may choose $\tau$} as the duration of the longest elementary operation, with shorter operations followed by wait locations as required~\cite[Sec.~15.5]{GottesmanQECC2024}. 
If $t_{\mathrm{store}}$ is the prescribed wall-clock storage time, then
\begin{equation}
    S=\left\lceil\frac{t_{\mathrm{store}}}{\tau}\right\rceil.
    \label{eq:main-cycle-count}
\end{equation}
Writing $\cN_\tau$ for the noise accumulated by one qubit during this interval, the channel used in the discrete model is $\cN=\cN_\tau$. 
A fixed-width memory using $N$ physical qubits therefore contains $NS$ storage locations and has physical qubit-time $\tau NS$. 
{We use the same time-step convention} for the unencoded and encoded memories.

\subsection{Erasure-noise lower bound}
\label{sec:methods-erasure-lower}
{Here we derive \cref{eq:main-survival,eq:main-erasure-fidelity,eq:main-erasure-one}; we give the complete proofs in \cref{lem:adaptive-survival,thm:erasure-one,prop:path-dimension,cor:erasure-additive} of \hyperref[sec:supp-erasure-lower]{\SM{}~\ref*{sec:supp-erasure-lower}}.}

At time step $s$, each of the $q_s(h_{s-1})$ physical qubits is independently erased with probability $p$, and every erasure location is revealed. The width is chosen before the erasures at the current time step are known. We call the event in which all $q_s(h_{s-1})$ qubits are erased a complete-erasure event.

For an integer bound $C$ on the circuit size of every individual execution, let $F(s,v)$ be the largest possible probability of avoiding a complete-erasure event during $s$ remaining time steps with $v$ physical storage locations available. If the next time step uses a live width $n\leq v$,  
complete erasure occurs with probability $p^n$, and the conditional probability of avoiding complete erasure thereafter is at most $F(s-1,v-n)$. Here $n$ is the width chosen for this time step and may change during the protocol; $N$ is reserved for the fixed block length in the upper-bound construction. 
Hence
\begin{equation}
F(s,v)\leq\max_{0\leq n\leq v}(1-p^n)F(s-1,v-n).
\end{equation}
Induction bounds an adaptive protocol by the optimal deterministic allocation $n_1+\cdots+n_S\leq C$. Applying $1-x\leq e^{-x}$ and Jensen's inequality to the convex function $p^x$ gives \cref{eq:main-survival}.

Conditioned on the first complete-erasure event, the inaccessible reference is separable from every system available to the recovery operation. Subsequent processing cannot restore this entanglement, so the final Bell-state fidelity on this event is at most $1/2$. Averaging over measurement and erasure records proves \cref{eq:main-erasure-fidelity}. Diamond-norm error $\eps$ implies $1-F_{\mathrm e}\leq\eps/2$ and hence \cref{eq:main-erasure-one}.

For $K$ logical qubits, use a reference system of dimension $2^K$. If every complete measurement record with nonzero probability contained a time step with fewer than $K$ physical qubits, then every conditional output state would have Schmidt number at most $2^{K-1}$ and overlap at most $1/2$ with the maximally entangled state. The same would hold for their mixture, contradicting diamond-norm error below one. Therefore at least one complete record satisfies $q_s\geq K$ at every time step, and $\Vpath\geq KS$. {We state the result} for finite or countable classical outcome sets; continuous outcomes require the corresponding measure-theoretic formulation. We give the complete proof in \cref{prop:path-dimension}.

\subsection{CSS-code upper bound}
\label{sec:methods-css-upper}

{Here we derive} the block-size choice in \cref{eq:main-block} and the upper bound in \cref{thm:tight-erasure}.

{Let $h_2$ denote the binary entropy and let $\delta_{\mathrm{GV}}\in(0,1/2)$ be the unique solution of $h_2(\delta_{\mathrm{GV}})=1/2$; numerically, $\delta_{\mathrm{GV}}\simeq0.1100$.} For fixed $p<\delta_{\mathrm{GV}}$, choose
\begin{equation}
\begin{gathered}
p<\Delta<\delta_{\mathrm{GV}},
\,\,\, 0<R<1-2h_2(\Delta),\,\,\,
a=D(\Delta\|p)>0.
\end{gathered}
\end{equation}
Here $D(x\|y)=x\log(x/y)+(1-x)\log((1-x)/(1-y))$ denotes the binary relative entropy, with natural logarithms.
For every sufficiently large integer $N$, the CSS Gilbert--Varshamov bound provides a code encoding at least $RN$ logical qubits with distance at least $\Delta N$. Let $N_0$ be an integer such that this code-existence statement holds for every $N\geq N_0$, and choose
\begin{equation}
N=\left\lceil\max\!\left\{N_0,\frac{K+1}{R},\frac1a\log\frac{2S}{\eps}\right\}\right\rceil.
\end{equation}
The number $X$ of erasures at one time step is distributed as $\mathrm{Bin}(N,p)$, and the Chernoff bound gives $\Pr[X\geq\Delta N]\leq e^{-aN}$. Every erasure pattern of weight below the code distance is exactly correctable. On the remaining patterns the recovered channel may be arbitrary, so the diamond-norm error of one time step with ideal recovery is at most $2e^{-aN}$. The telescoping inequality over $S$ time steps gives total error at most $2Se^{-aN}\leq\eps$. Since the protocol has fixed width, $\Vpath=NS$. The construction is nonconstructive and permits inefficient, nonlocal ideal recovery. The restriction $p<\delta_{\mathrm{GV}}$ applies only to this upper-bound construction. We give the code-existence bound, erasure tail bound, block construction, and tightness statement in \cref{thm:css-gv,lem:chernoff,thm:wide-erasure,cor:tight-erasure} of \hyperref[sec:supp-css-upper]{\SM{}~\ref*{sec:supp-css-upper}}.

\subsection{Circuit-level upper bound}
\label{sec:methods-circuit-upper}

We state the formal gadget assumptions in \cref{ass:gadgets} of \hyperref[sec:supp-circuit-upper]{\SM{}~\ref*{sec:supp-circuit-upper}}.
The circuit theorem assumes constants $R_{\mathrm{gad}},c_q,c_t,\beta>0$, $A\geq1$, $c_N\geq1$, and $N_0$. For every sufficiently large target length, an available block length lies within a factor $c_N$ of that target. A length-$N$ block encodes at least $R_{\mathrm{gad}}N$ logical qubits. Every allowed logical layer, including state preparation and measurement, has a gadget using at most $c_qN$ physical qubits for at most $c_t$ time steps. For every preceding measurement record and every reference system, the joint quantum--classical output channel is within diamond-norm distance $Ae^{-\beta N}$ of the ideal logical layer. Every quantum system and classical output passed between successive gadgets is included in this condition and in the resource count.

Set %
\begin{equation}
x=\max\!\left\{N_0,\frac{K}{R_{\mathrm{gad}}},
\frac1\beta\log\frac{AT}{\eps}\right\},
\end{equation}
and choose an available block length $N$ satisfying $x\leq N\leq c_Nx$.
The block encodes all $K$ logical qubits, and the telescoping inequality over $T$ adaptive logical layers gives total diamond-norm error at most $TAe^{-\beta N}\leq\eps$. The implementation depth is at most $c_tT$, and $\CFT\leq c_qc_tNT$. This proves \cref{eq:main-circuit}.
We give the formal statement and proof in \cref{thm:circuit-achievability} of \hyperref[sec:supp-circuit-upper]{\SM{}~\ref*{sec:supp-circuit-upper}}. %

\subsection{Subsystem spacetime-code bounds}
\label{sec:methods-spacetime}

{We define an exact subsystem spacetime code} only after selecting elementary Pauli fault locations and specifying a circuit-to-code map. Its coordinates need not be in one-to-one correspondence with physical storage locations. Here $\Cmap$ denotes the location count associated with that map; it need not equal $\Vpath$ of the memory model. To obtain a circuit-size bound, we assume that $\Nst\leq\cloc\Cmap$ for a scale-independent constant $\cloc$. Applying the binary subsystem Singleton bound to $[[\Nst,k,r,\dst]]$ proves \cref{eq:main-singleton,eq:main-volume}; {we state the formal assumption and give the proof in \cref{ass:coordinate-volume,thm:st-singleton} of \hyperref[sec:supp-spacetime]{\SM{}~\ref*{sec:supp-spacetime}}.}

For the probability bound, Pauli fault strings have distribution $\nu$, and $\PL$ denotes the probability that recovery leaves a nontrivial dressed logical Pauli. The decoder receives the complete stabilizer syndrome and independent private randomness and returns a Pauli recovery with the same syndrome. For a dressed logical Pauli $L$, the map $F\mapsto FL$ partitions the faults into pairs with identical syndrome. The decoder cannot correct both faults in a pair. Summing the smaller probability in each pair gives \cref{eq:main-pairing}. Under product noise, $F$ and $FL$ differ only on $\supp(L)$, and each single-coordinate likelihood ratio lies in $[\rho,\rho^{-1}]$. Taking $\wt(L)=\dst$ gives \cref{eq:main-pair-product}; rearranging this bound and applying the Singleton inequality gives \cref{eq:main-distance}. { We give the complete proofs in \cref{thm:pauli-pair,cor:accuracy-distance} of \hyperref[sec:supp-spacetime]{\SM{}~\ref*{sec:supp-spacetime}}, while the separate sufficient condition for a positive threshold is \cref{thm:witness-threshold} in \hyperref[sec:supp-threshold]{\SM{}~\ref*{sec:supp-threshold}}.}

\makeatletter
\renewcommand{\p@subsection}{\thesection\,}
\makeatother

\renewcommand{\thesection}{\Alph{section}}
\renewcommand{\thesubsection}{\arabic{subsection}}

\section*{Data availability}

No datasets were generated or analysed in this study.

\section*{Code availability}

No custom code was used to establish the results. All analytical derivations are provided in the \SM{}.

\section*{Acknowledgements}

A.T. is supported by the CQT PhD Scholarship, the Google PhD Fellowship Program and the CQT Young Researcher Career Development Grant. K.B.~is supported by a Hartree Fellowship from the Joint Center for Quantum Information and Computer Science (QuICS) at the University of Maryland, College Park. Generative artificial intelligence tools were used to assist with ideation, language editing, organization, and the preparation of portions of the manuscript.

\section*{Author contributions}
All authors contributed equally to the conceptualization, methodology,
formal analysis, verification, and interpretation of the results. All
authors contributed equally to drafting, reviewing, and editing the
manuscript and approved its final version.

\section*{Competing interests}

The authors declare no competing interests.

\bibliography{references}

\clearpage
\newpage

\let\addcontentsline\oldaddcontentsline

\renewcommand{\appendixname}{\SMLong{}}
\appendix

\onecolumngrid
\newpage

\setcounter{secnumdepth}{2}
\setcounter{equation}{0}
\setcounter{figure}{0}
\setcounter{section}{0}

\renewcommand{\thesection}{\Alph{section}}
\renewcommand{\thesubsection}{\arabic{subsection}}
\renewcommand*{\theHsection}{\thesection}

\clearpage
\begin{center}

\textbf{\large \SMLong{}}
\end{center}
\setcounter{equation}{0}
\setcounter{figure}{0}
\setcounter{table}{0}

\makeatletter

\renewcommand{\thefigure}{S\arabic{figure}}

In the \SMLong{}, we provide proofs and additional details supporting the claims in the main text.

\makeatletter
\@starttoc{toc}

\makeatother

\section{Notation}
\label{sec:supp-notation}
{We summarize the symbols used in this paper in} \cref{tab:notation}.
{Constants hidden by asymptotic notation may depend on parameters explicitly fixed in the corresponding statement, but not on the scaling variables.}
{In the uniform erasure bounds, these constants may depend on the fixed erasure rate $p$ and on $\eps_0$, but not on $K$, $S$, or $\eps$.}

\begin{table*}[h]
\caption{Symbols.}
\label{tab:notation}
\begin{ruledtabular}
\begin{tabular}{
p{0.075\textwidth}
p{0.36\textwidth}
p{0.085\textwidth}
p{0.36\textwidth}
}
Symbol & Meaning & Symbol & Meaning \\
\hline

$K$
&
Logical qubits stored together in the memory and circuit results.
&
$k,r$
&
Protected and gauge qubits of a subsystem spacetime code.
\\

$T,S$
&
Ideal circuit depth and implementation depth (number of time steps),
respectively; these need not be equal.
&
$\partial_s$
&
The $s$th application of the physical noise channel.
\\

$q_s(h_{s-1})$
&
Number of physical qubits immediately before $\partial_s$, selected from
the preceding measurement record.
&
$\overline q_s$
&
Maximum physical width at time step $s$ over preceding measurement
records.
\\

$\Qpeak$
&
Peak physical width over time steps and measurement records.
&
$\Videal$
&
Ideal logical circuit size, counting one wire location per logical qubit
and layer.
\\

$\mathcal Q$
&
Adaptive quantum circuit implementing the memory protocol.
&
$\CFT,\Cmap$
&
Worst-case physical count in the circuit theorem and location count associated with the spacetime circuit-to-code map, respectively.
\\

$\Venv$
&
Sum over time steps of the maximum physical width at each time step.
&
$\Vpath$
&
Worst-case storage-location count over complete measurement records for one protocol.
\\

$\Vav$
&
Expected physical circuit size, counting storage locations.
&
$\Nst$
&
Selected Pauli fault coordinates in an exact spacetime code.
\\

$\dst$
&
Minimum weight of a nontrivial dressed logical Pauli.
&
$\cloc$
&
Scale-independent constant in
$\Nst\leq\cloc\Cmap$ for the specified spacetime circuit-to-code
measure.
\\

$\PL$
&
Logical Pauli failure probability after complete-syndrome Pauli recovery.
&
$\Ov_\nu(L)$
&
Overlap between the fault distributions $\nu(F)$ and $\nu(FL)$.
\\

$\qmin,\qmax$
&
Extrema over selected coordinates and all four Pauli-symbol probabilities,
including $I$.
&
$\rho$
&
Ratio $\qmin/\qmax$.
\\

$\eps$
&
Variable target diamond error or, where stated, target logical Pauli
failure.
&
$\delta$
&
Fixed target diamond error in the general-channel theorem.
\\

$p$
&
Erasure or physical fault probability.
&
$\kappa_{\cN}$
&
Contraction constant for the fixed non-unitary qubit channel $\cN$.
\\

$\delta_{\mathrm{GV}}$
&
Smaller root of $h_2(x)=1/2$, approximately $0.1100$.
&
$R_{\mathrm{gad}}$
&
Rate lower bound for the code family in the conditional circuit theorem.
\\

\end{tabular}
\end{ruledtabular}
\end{table*}

\section{Preliminaries}

\subsection{States and channels}

A density operator $\rho$ is positive semidefinite with $\Tr\rho=1$. A quantum channel is a completely positive trace-preserving linear map. For a linear map $\mathcal T$, the diamond norm is
\begin{equation}
\|\mathcal T\|_\diamond=
\sup_{m\geq1}\sup_{X\neq0}
\frac{\|(\mathcal T\otimes\id_m)(X)\|_1}{\|X\|_1}.
\end{equation}
The auxiliary dimension need not exceed the input dimension. Every channel has diamond norm one, the norm is stable under tensoring with an identity channel, and it is submultiplicative under composition. If $\|\cM-\id\|_\diamond\leq\eps$, then this bound applies in particular to an input entangled with an inaccessible reference.

For an integer $D\geq2$, let $|\Phi_D\rangle$ denote the maximally entangled state of two $D$-dimensional systems:
\begin{equation}
|\Phi_D\rangle=D^{-1/2}\sum_{j=0}^{D-1}|j\rangle|j\rangle,
\qquad \Phi_D=|\Phi_D\rangle\!\langle\Phi_D|.
\end{equation}
For $D=2$ we write $\Phi$. A bipartite state is separable if it is a convex combination of product states.

\begin{lemma}[Bell-state witness]\label{lem:bell-witness}
For every separable two-qubit state $\sigma$,
\begin{equation}
\Tr(\Phi\sigma)\leq\frac12,
\qquad
\|\Phi-\sigma\|_1\geq1.
\end{equation}
\end{lemma}

\begin{proof}
For a pure product state $|a\rangle|b\rangle$, Cauchy--Schwarz gives
\begin{equation}
|\langle\Phi|a,b\rangle|^2
=\frac12|a_0b_0+a_1b_1|^2\leq\frac12.
\end{equation}
Linearity extends the fidelity bound to mixtures. Measuring the two-outcome test $\{\Phi,I-\Phi\}$ distinguishes $\Phi$ from $\sigma$ with classical $\ell_1$ distance at least $2(1-\Tr\Phi\sigma)\geq1$. Trace norm cannot increase under a measurement, so $\|\Phi-\sigma\|_1\geq1$.
\end{proof}

We will use the telescoping identity. For implemented channels $\Lambda_1,\ldots,\Lambda_T$ and their corresponding ideal channels $\Gamma_1,\ldots,\Gamma_T$, %
\begin{align}
&\Lambda_T\cdots\Lambda_1-\Gamma_T\cdots\Gamma_1=\sum_{t=1}^{T}\Lambda_T\cdots\Lambda_{t+1}
(\Lambda_t-\Gamma_t){\circ}\Gamma_{t-1}\cdots\Gamma_1.
\end{align}
Submultiplicativity and unit diamond norm of channels imply
\begin{align}
    &\|\Lambda_T\cdots\Lambda_1-
    \Gamma_T\cdots\Gamma_1\|_\diamond
    \leq\sum_{t=1}^{T}
    \|\Lambda_t-\Gamma_t\|_\diamond,
\label{eq:supp-telescope}
\end{align}
provided each bound is uniform over the reference systems and classical conditioning needed at that replacement.

\subsection{Schmidt number}

For a bipartite pure state on $A:B$, its Schmidt rank is the number of nonzero coefficients in its Schmidt decomposition across the bipartition $A:B$~\cite[Sec.~2.5]{NielsenChuang2010}.
The Schmidt number of a mixed state is the smallest integer $d$ such that the state can be written as a convex combination of pure states, each of Schmidt rank at most $d$. Local channels and conditioning on local classical outcomes do not increase Schmidt number. 
\begin{lemma}[Singlet fraction from Schmidt number]\label{lem:singlet-schmidt}
If a state $\sigma$ has Schmidt number at most $d$, then
\begin{equation}
\Tr(\Phi_D\sigma)\leq\frac dD.
\end{equation}
\end{lemma}

\begin{proof}
Consider first a pure state $|\psi\rangle$ of Schmidt rank $r\leq d$. In matrix form, its coefficient matrix $A$ has Frobenius norm one and rank at most $r$. Since $\langle\Phi_D|\psi\rangle=D^{-1/2}\Tr A$, Cauchy--Schwarz for the nonzero singular values gives $|\Tr A|\leq\sqrt r\|A\|_2=\sqrt r$. Thus $|\langle\Phi_D|\psi\rangle|^2\leq r/D\leq d/D$. A convex decomposition into such pure states and linearity prove the mixed-state claim.
\end{proof}

\subsection{Paulis and subsystem codes}

We regard Pauli strings modulo global phase. A binary linear subsystem stabilizer code $[[N,k,r,d]]$ has stabilizer group $\mathcal S$, gauge group $\mathcal G$, $k$ protected logical qubits and $r$ gauge qubits. Its dressed distance, which allows a logical Pauli to act on the gauge subsystem and minimizes its weight modulo gauge operators, is %
\begin{equation}
d=\min\{\wt(L):L\in N(\mathcal S)\setminus\mathcal G\}.
\end{equation}
The binary subsystem Singleton inequality is
\begin{equation}
k+r\leq N-2(d-1)
\label{eq:supp-singleton-base}
\end{equation}
for $d\geq2$, including impure binary linear subsystem codes~\cite{Klappenecker2007}. When $d=1$, the same displayed inequality reduces to the elementary dimension bound $k+r\leq N$.

Here a coordinate is one of the $N$ qubit positions, and a Pauli distribution assigns a probability to each Pauli string $F=(F_1,\ldots,F_N)$. A full-support product Pauli distribution has 
\begin{equation}
\nu(F)=\prod_{i=1}^{N}\pi_i(F_i),\qquad
\pi_i(P)>0\quad(P\in\{I,X,Y,Z\}).
\end{equation}
A random Pauli is local stochastic with rate $p$ if, for every coordinate set $A$, the probability that all coordinates in $A$ are faulty is at most $p^{|A|}$. Independent erasure noise with rate $p$ replaces each erased qubit by an orthogonal erasure state and reveals the erased locations to the recovery operation.

\subsection{The entanglement functional}

For $\supp(\rho)\subseteq\supp(\sigma)$ define
\begin{equation}
\chi_2(\rho,\sigma)=\Tr[(\rho-\sigma)\sigma^{-1/2}(\rho-\sigma)\sigma^{-1/2}],
\end{equation}
with the inverse restricted to the support of $\sigma$; otherwise set $\chi_2=+\infty$. Define %
\begin{equation}
\Echi(\rho_{AB})=\min_{\sigma\in\Sep(A:B)}\chi_2(\rho_{AB},\sigma).
\end{equation}
Thus $\Echi$ is the smallest $\chi_2$-divergence from $\rho_{AB}$ to a separable state across $A:B$; we use it to track the loss of entanglement under noise.
{We use} four standard properties {in the proof}~\cite{FMHS2022}: $\|\rho-\sigma\|_1^2\leq\chi_2(\rho,\sigma)$; monotonicity of $\Echi$ under separable channels; convexity; and $\Echi(\Phi)\leq3$. For the last claim choose $\sigma=I/4$ and compute
\begin{equation}
\chi_2(\Phi,I/4)=4\Tr[(\Phi-I/4)^2]=3.
\end{equation}
Consequently, with $d_{\Sep}(\rho)=\min_{\sigma\in\Sep}\|\rho-\sigma\|_1$,
\begin{equation}
d_{\Sep}(\rho)^2\leq\Echi(\rho).
\label{eq:supp-sep-chi}
\end{equation}

\section{Quantum-memory model}
\label{sec:supp-memory-model}

\subsection{Ideal and fault-tolerant circuits}

Consider an ideal noise-free circuit $\mathcal C=\mathcal C_T\circ\cdots\circ\mathcal C_1$. We count one logical circuit location for each logical qubit present in each time step, whether it undergoes a nontrivial gate or an identity.
If $n_t$ logical qubits are present after layer $t$, its peak width and circuit size are
\begin{equation}
n=\max_t n_t,\qquad \Videal=\sum_{t=1}^{T}n_t.
\end{equation}
For a constant-width $K$-qubit memory, the ideal operation is the identity, but the prescribed duration gives $T$ layers of $K$ storage locations and hence $\Videal=KT$.

The physical protocol $\mathcal Q$ evolves for $S$ time steps. Immediately before time step $s$, the preceding classical measurement record $h_{s-1}$ determines the complete physical register $\Gamma_s(h_{s-1})$ and its width $q_s(h_{s-1})$; the specified noise channel is then applied once to every qubit in that register, and this application is denoted by $\partial_s$. The protocol may use earlier outcomes to choose the register at time $s$, but it cannot choose the time-step schedule after seeing later noise outcomes.
The register $\Gamma_s(h_{s-1})$ includes every quantum system retained for use after the noise application; there is no uncounted noiseless quantum memory.
A physical qubit stored during one time step contributes one storage location, including when the same physical qubit 
is reused at later time steps. 
{Thus, for a complete record $h_S$, we denote the location-count circuit size by
$C(\mathcal Q;h_S)=\sum_{s=1}^{S}q_s(h_{s-1})$.
See \cref{fig:memory-protocols} for a pictorial synopsis. In the fixed-width
case, $q_s=N$ for every $s$, so $C(\mathcal Q;h_S)=NS$; in the adaptive
case, $q_s$ is determined by the prior record $h_{s-1}$, so the execution
cost can depend on the realized measurement record.}

Only systems present at a modeled time step enter the storage-location count. A circuit-level model additionally includes the preparation, gate, measurement, and wait locations used between successive time steps. Classical measurement records are not counted as quantum circuit locations.
{Thus an ancilla prepared and discarded entirely between two modeled time steps does not contribute to the storage-only count, whereas an ancilla retained until a modeled time step is included in $q_s$.}

\subsection{Circuit-size measures}

We use four resource measures: peak width, the sum of the largest width at each time step, the worst-case size of a complete execution, and the mean execution size. %
Define
\begin{align}
\Qpeak&=\max_{1\leq s\leq S}\sup_{h_{s-1}}q_s(h_{s-1}),\\
\Venv&=\sum_{s=1}^{S}\sup_{h_{s-1}}q_s(h_{s-1}),\\
\Vpath(\mathcal Q)&=\sup_{h_S}\sum_{s=1}^{S}q_s(h_{s-1}),\\
\Vav(\mathcal Q)&=\mathbb E_{h_S}\sum_{s=1}^{S}q_s(h_{s-1}).
\end{align}
Then
\begin{equation}
\Vav(\mathcal Q)\leq\Vpath(\mathcal Q)\leq\Venv\leq\Qpeak S.
\label{eq:supp-volume-chain}
\end{equation}
The first inequality compares an expectation with a supremum. The second holds because the width of any execution at time step $s$ is at most the maximum width $\sup_{h_{s-1}}q_s(h_{s-1})$ at that time step. The third follows from the definition of the peak width $\Qpeak$. The two circuit-size measures can differ: if a random bit selects width 100 at time step 1 or width 100 at time step 2, but never both, then every execution has size 100 whereas $\Venv=200$.

\subsection{Circuit locations}

In the standard threshold-theorem formulation, a circuit is decomposed into preparation, gate, measurement, and wait or storage locations, and the noise model assigns an error rate to each type of location. A fault-tolerant simulation then replaces every ideal location, including a storage location, by the corresponding encoded gadget. 
Here an encoded gadget is the physical fault-tolerant circuit that implements one ideal logical location on encoded data, including the required error correction.
{We therefore define circuit-size overhead as} the ratio of physical to ideal location counts~\cite[Ch.~10]{GottesmanQECC2024}. This accounting treats active operations and storage within a single circuit model.

For a quantum memory, a physical qubit present during one time step is a storage location. Hence $\sum_s q_s(h_{s-1})$ is the circuit size of the execution specified by $h_S$, $\Vpath$ is the worst-case circuit size over complete measurement records, and $\Vav$ is the expected circuit size. A gate-level implementation refines each modeled time step into preparation, recovery, routing, measurement, and wait locations, so its total location count is at least the storage-location count used here. If each modeled time step is represented by a bounded number of elementary layers and all participating ancillas are included, the two location counts agree up to constant factors. The location count also records the duration of an identity circuit, whose nonidentity logical-gate count is zero. Accordingly, the general-channel theorem lower-bounds $\Venv$, the sum of the maximum widths over the individual time steps, whereas the erasure theorem lower-bounds $\Vpath$, the worst-case size of a realizable execution.

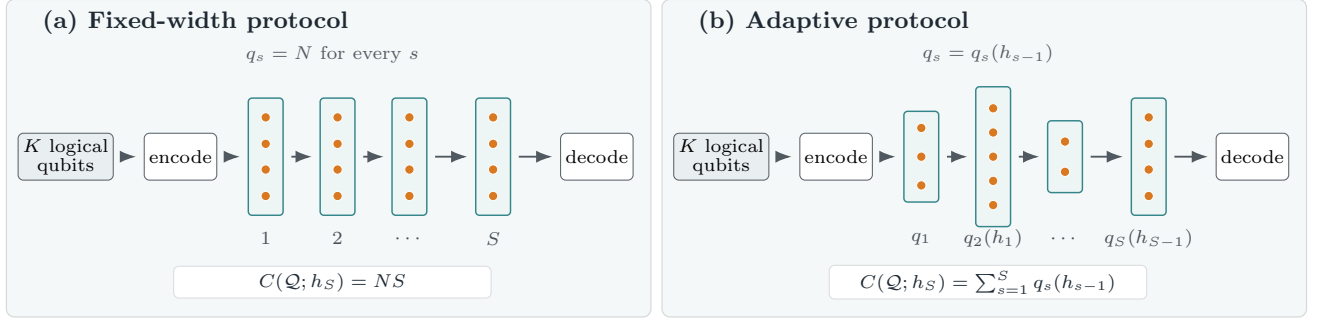
\begin{figure*}[t]
\centering
\begin{tikzpicture}[
  x=0.98cm,y=1cm,>={Latex[length=1.8mm,width=1.3mm]},font=\scriptsize,
  panel/.style={draw=natureink!18,fill=naturepale,rounded corners=3pt,
    line width=0.45pt},
  io/.style={draw=natureink!72,rounded corners=2pt,fill=naturegrey,
    minimum width=1.08cm,minimum height=0.62cm,align=center,inner sep=1.8pt,
    font=\scriptsize},
  control/.style={draw=natureink!72,rounded corners=2pt,fill=white,
    minimum width=0.90cm,minimum height=0.62cm,align=center,inner sep=1.8pt,
    font=\scriptsize},
  flow/.style={-Latex,line width=0.70pt,natureink!78,
    shorten >=0pt,shorten <=0pt},
  slice/.style={draw=natureblue!82,fill=natureblue!7,
    rounded corners=1.4pt,line width=0.55pt,minimum width=0.46cm,
    inner sep=0pt},
  store/.style={circle,fill=natureorange,draw=white,line width=0.25pt,
    inner sep=0pt,minimum size=3.5pt},
  formula/.style={draw=natureink!16,fill=white,rounded corners=2pt,
    minimum width=4.20cm,minimum height=0.42cm,inner sep=2pt}
]
  \begin{scope}
    \path[panel] (-0.15,-0.62) rectangle (8.55,3.58);
    \node[anchor=west,font=\small\bfseries,natureink] at (0.22,3.28)
      {{(a) Fixed-width protocol}};
    \node[anchor=center,natureink!76] at (4.25,2.86)
      {{$q_s=N$ for every $s$}};

    \node[io]      (fin)  at (0.65,1.50)
      {{$K$ logical}\\[-1pt]{qubits}};
    \node[control] (fenc) at (2.20,1.50) {{encode}};
    \node[slice,minimum height=1.55cm] (f1) at (3.35,1.50) {};
    \node[slice,minimum height=1.55cm] (f2) at (4.32,1.50) {};
    \node[slice,minimum height=1.55cm] (f3) at (5.29,1.50) {};
    \node[slice,minimum height=1.55cm] (fS) at (6.42,1.50) {};
    \node[control] (fdec) at (7.82,1.50) {{decode}};

    \draw[flow] ([xshift=1mm]fin.east)--([xshift=-1mm]fenc.west);
    \draw[flow] ([xshift=1mm]fenc.east)--([xshift=-1mm]f1.west);
    \draw[flow] ([xshift=1mm]f1.east)--([xshift=-1mm]f2.west);
    \draw[flow] ([xshift=1mm]f2.east)--([xshift=-1mm]f3.west);
    \draw[flow] ([xshift=1mm]f3.east)--([xshift=-1mm]fS.west);
    \draw[flow] ([xshift=1mm]fS.east)--([xshift=-1mm]fdec.west);

    \foreach \x in {3.35,4.32,5.29,6.42}
      \foreach \y in {0.98,1.33,1.67,2.02}
        \node[store] at (\x,\y) {};

    \node[natureink!82] at (3.35,0.42) {{$1$}};
    \node[natureink!82] at (4.32,0.42) {{$2$}};
    \node[natureink!82] at (5.29,0.42) {{$\cdots$}};
    \node[natureink!82] at (6.42,0.42) {{$S$}};
    \node[formula,text=natureink] at (4.25,-0.16)
      {{$C(\mathcal Q;h_S)=NS$}};
  \end{scope}

  \begin{scope}[shift={(8.85,0)}]
    \path[panel] (-0.15,-0.62) rectangle (8.55,3.58);
    \node[anchor=west,font=\small\bfseries,natureink] at (0.22,3.28)
      {{(b) Adaptive protocol}};
    \node[anchor=center,natureink!76] at (4.25,2.86)
      {{$q_s=q_s(h_{s-1})$}};

    \node[io]      (ain)  at (0.65,1.50)
      {{$K$ logical}\\[-1pt]{qubits}};
    \node[control] (aenc) at (2.20,1.50) {{encode}};
    \node[slice,minimum height=1.20cm] (a1) at (3.35,1.50) {};
    \node[slice,minimum height=1.85cm] (a2) at (4.32,1.50) {};
    \node[slice,minimum height=0.95cm] (a3) at (5.29,1.50) {};
    \node[slice,minimum height=1.55cm] (aS) at (6.42,1.50) {};
    \node[control] (adec) at (7.82,1.50) {{decode}};

    \draw[flow] ([xshift=1mm]ain.east)--([xshift=-1mm]aenc.west);
    \draw[flow] ([xshift=1mm]aenc.east)--([xshift=-1mm]a1.west);
    \draw[flow] ([xshift=1mm]a1.east)--([xshift=-1mm]a2.west);
    \draw[flow] ([xshift=1mm]a2.east)--([xshift=-1mm]a3.west);
    \draw[flow] ([xshift=1mm]a3.east)--([xshift=-1mm]aS.west);
    \draw[flow] ([xshift=1mm]aS.east)--([xshift=-1mm]adec.west);

    \foreach \y in {1.12,1.50,1.88} \node[store] at (3.35,\y) {};
    \foreach \y in {0.86,1.18,1.50,1.82,2.14}
      \node[store] at (4.32,\y) {};
    \foreach \y in {1.30,1.70} \node[store] at (5.29,\y) {};
    \foreach \y in {0.98,1.33,1.67,2.02}
      \node[store] at (6.42,\y) {};

    \node[natureink!82] at (3.35,0.42) {{$q_1$}};
    \node[natureink!82] at (4.32,0.42) {{$q_2(h_1)$}};
    \node[natureink!82] at (5.29,0.42) {{$\cdots$}};
    \node[natureink!82] at (6.42,0.42) {{$q_S(h_{S-1})$}};
    \node[formula,text=natureink] at (4.25,-0.16)
      {{$C(\mathcal Q;h_S)=\sum_{s=1}^{S}q_s(h_{s-1})$}};
  \end{scope}
\end{tikzpicture}
\caption{Quantum-memory circuit-size accounting. (a) In a fixed-width protocol, the live register contains $N$ physical qubits at every time step, so one execution contains $NS$ storage locations. (b) In an adaptive protocol, the earlier measurement record $h_{s-1}$ determines the live width $q_s(h_{s-1})$ at time step $s$. The orange dots denote storage locations. Ideal control between time steps may include recovery, measurement, fresh ancillas and classical feedforward. %
}
\label{fig:memory-protocols}
\end{figure*}

\subsection{Fixed-duration memories}
\label{sec:supp-fixed-duration}
A fixed-duration quantum memory may be adaptive: it receives an unknown $K$-qubit state initially, returns it after $S$ time steps, and may choose its width at time $s$ from the classical record of earlier time steps.
The protocol may not replace the prescribed storage interval by a shorter identity channel or transfer input-dependent quantum information to an uncounted noiseless memory before the final time. This additional timing requirement is necessary because identity circuits of different depths implement the same input--output channel.

For a family indexed by a code-family parameter %
$\lambda$, a decoder threshold is
\begin{equation}
\begin{aligned}
    p_{\mathrm{th}}=\sup\bigl\{p\geq0:\;&
    \text{for every }0\leq p'<p,
    \lim_{\lambda\to\infty}\PL(\lambda,p')=0\bigr\}.
\end{aligned}
\label{eq:supp-decoder-threshold}
\end{equation}
The family has a positive decoder threshold when $p_{\mathrm{th}}>0$.
The parameter $\lambda$ may be the block length, lattice size, or concatenation level and need not equal the code distance. The threshold condition states that the logical error can be made arbitrarily small by increasing $\lambda$. %
The resource lower bound proved here applies to every protocol that achieves a prescribed logical error.

\section{Circuit-size lower bound for general noise}
\label{sec:supp-general-noise}

\subsection{Entanglement contraction}

\begin{lemma}[Entanglement contraction for $m$ noisy qubits]\label{lem:fmhs-exact}
For every non-unitary qubit channel $\cN$, there is $\kappa_{\cN}\in(0,1]$ such that the following holds. Let a bipartite register contain $m\geq1$ noisy qubits and finite-dimensional classical registers. Apply $\cN^{\otimes m}$ to the qubits and then a separable channel across the bipartition. If $\rho$ and $\rho'$ are the states before and after these operations, then
\begin{equation}
\Echi(\rho')\leq(1-\kappa_{\cN}^{m})\Echi(\rho).
\label{eq:supp-fmhs}
\end{equation}
\end{lemma}

\begin{proof}
The revised proof of Lemma~5 in Ref.~\cite{FMHS2022} establishes the inequality above with $m=n_i$, the actual number of noisy qubits in the layer. Equations~(12)--(18) give the unital-channel case, while Eqs.~(19) and (27) extend the argument to non-unital qubit channels. We use this layer-dependent form before replacing $m$ by a global peak width. The classical registers are block-diagonal outcome registers and do not represent coherent quantum side information.
\end{proof}

\subsection{Layerwise circuit-size lower bound}
We now accumulate the entanglement contraction over the storage time. This gives a lower bound on $C_{\mathrm{lay}}$, the sum of the largest physical widths at the individual time steps, for any non-unitary qubit channel.
\begin{theorem}[One-qubit lower bound]\label{thm:integrated-envelope}
Let $\cN$ be a non-unitary qubit channel and $\kappa=\kappa_{\cN}$. Let a memory channel $\cM$ store one qubit for $S$ time steps, with $\overline q_s=\sup_{h_{s-1}}q_s(h_{s-1})$ and $\Venv=\sum_s\overline q_s$. If
\begin{equation}
\|\cM-\id\|_\diamond\leq\delta<\frac12,
\end{equation}
then
\begin{equation}
(1-2\delta)^2\leq3\prod_{s=1}^{S}(1-\kappa^{2\overline q_s}).
\label{eq:supp-product}
\end{equation}
Let $C_\delta=\log[3/(1-2\delta)^2]$. If $0<\kappa<1$ and $S>C_\delta$, then
\begin{equation}
\Venv\geq\frac{S}{2\log(1/\kappa)}\log\frac{S}{C_\delta}.
\label{eq:supp-integrated}
\end{equation}
If $\kappa=1$, no memory with error $\delta<1/2$ survives one nonempty time step.
\end{theorem}

\begin{proof}
\emph{Bell-pair test.} Prepare $\Phi$ and send one half through each of two independent copies of the memory. Label the two memory copies %
$A$ and $B$.
{Both halves of the Bell pair therefore pass through the physical noise at every time step, as required by the contraction argument.}
Their encoders are local across $A:B$, so monotonicity and \cref{eq:supp-sep-chi} give $\Echi(\rho_0)\leq3$.

\emph{Contraction over successive time steps.} 
We may assume that every $\overline q_s$ is finite, since otherwise the circuit-size bound is immediate. Because the widths are integer-valued, each finite supremum is attained. Also $\overline q_s\geq1$: if the physical register were empty on every branch before some noise application, only classical data and newly prepared local systems would remain, so the decoded state would be separable. By \cref{lem:bell-witness}, its distance from $\Phi$ would then be at least one, contradicting the two-copy error bound $2\delta<1$.

For a measurement record using fewer than $\overline q_s$ physical qubits, add fixed local qubits, apply the same noise to them, and discard them. This addition changes neither the logical channel nor $\Echi$. The two copies then contain exactly $2\overline q_s$ noisy qubits. Their control operations, including measurement and classical feedforward, form a separable channel across $A:B$. Iterating \cref{lem:fmhs-exact} yields
\begin{equation}
\Echi(\rho_S)\leq3\prod_{s=1}^{S}(1-\kappa^{2\overline q_s}).
\label{eq:supp-contracted}
\end{equation}
Local decoders do not increase this upper bound.

\emph{Accuracy preserves entanglement.} The decoded state is $\omega=(\cM\otimes\cM)(\Phi)$. Since
\begin{equation}
\cM\otimes\cM-\id\otimes\id
=(\cM-\id)\otimes\cM+\id\otimes(\cM-\id),
\end{equation}
we have $\|\omega-\Phi\|_1\leq2\delta$. For every separable $\sigma$, the triangle inequality and \cref{lem:bell-witness} give $\|\omega-\sigma\|_1\geq1-2\delta$. Hence \cref{eq:supp-sep-chi} implies $\Echi(\omega)\geq(1-2\delta)^2$. Combining with \cref{eq:supp-contracted} proves \cref{eq:supp-product}.

\emph{From product to sum.} The inequality $1-x\leq e^{-x}$ gives
\begin{equation}
\prod_s(1-\kappa^{2\overline q_s})
\leq\exp\!\left[-\sum_s\kappa^{2\overline q_s}\right].
\end{equation}
Thus \cref{eq:supp-product} implies
\begin{equation}
\sum_{s=1}^{S}\kappa^{2\overline q_s}\leq C_\delta.
\label{eq:supp-kappa-sum}
\end{equation}

\emph{From the width sum to circuit size.} When $0<\kappa<1$, $f(x)=\kappa^{2x}$ is convex. Jensen gives
\begin{equation}
\sum_s\kappa^{2\overline q_s}\geq
S\kappa^{2\Venv/S}.
\end{equation}
Together with \cref{eq:supp-kappa-sum}, taking logarithms and using $\log\kappa<0$ gives \cref{eq:supp-integrated}. If $\kappa=1$, every factor associated with a nonempty time step in \cref{eq:supp-product} vanishes, contradicting $\delta<1/2$.
\end{proof}

\subsection{Dimension bound and combined result}

\begin{proposition}[Information-dimension lower bound]\label{prop:envelope-dimension}
If a protocol stores $K\geq1$ logical qubits through $S$ time steps and $\|\cM_K-\id_K\|_\diamond\leq\delta<1$, then $\overline q_s\geq K$ at every time step and $\Venv\geq KS$.
\end{proposition}

\begin{proof}
Let $D=2^K$ and feed one half of $\Phi_D$ into the memory. At boundary $s$, include the complete classical transcript and pad each branch to $\overline q_s$. The live quantum register has dimension $2^{\overline q_s}$, so its state with the reference has Schmidt number at most $2^{\overline q_s}$. Later local processing cannot increase it. By \cref{lem:singlet-schmidt}, the final overlap with $\Phi_D$ is at most $2^{\overline q_s-K}$. Diamond closeness gives overlap at least $1-\delta/2>1/2$. Hence $2^{\overline q_s-K}>1/2$. Since $\overline q_s-K$ is an integer, $\overline q_s\geq K$. Summing over $s$ proves the claim.
\end{proof}

\begin{theorem}[Layerwise circuit-size lower bound for general noise]\label{cor:general-tradeoff}
Let $0<\kappa_{\cN}<1$. Every $K$-qubit memory lasting $S$ time steps with fixed diamond-norm error $\delta<1/2$ and $S>C_\delta$ obeys
\begin{equation}
\begin{aligned}
\Venv&\geq\max\!\left\{KS,
\frac{S\log(S/C_\delta)}{2\log(1/\kappa_{\cN})}\right\},\\
\Venv&={\Omega\!\left(S\left(K+\log S\right)\right)}.
\end{aligned}
\end{equation}
\end{theorem}

\begin{proof}
The first term follows from \cref{prop:envelope-dimension}. For the second, fix $K-1$ logical inputs to pure states and discard their outputs. Pre- and post-processing do not increase the diamond-norm error, so this gives a one-qubit memory with no larger value of $\Venv$. Apply \cref{thm:integrated-envelope} and use $\max\{x,y\}\geq(x+y)/2$.
\end{proof}

\begin{corollary}[No uniform constant relative bound on $\Venv$]\label{cor:no-uniform}
Under the fixed-duration assumptions in \hyperref[sec:supp-fixed-duration]{\SM{}~\ref*{sec:supp-fixed-duration}} and the clock convention in \hyperref[sec:methods-time-step]{Methods~\ref*{sec:methods-time-step}}, no compiler for a non-unitary qubit channel can satisfy $\Venv\leq C\Videal$ with one constant $C$ for every memory width and duration at any fixed diamond-norm error below $1/2$.
\end{corollary}

\begin{proof}
Choose a one-qubit ideal memory of duration $T$, so $\Videal=T$. The fixed-duration and clock assumptions give $S\geq T$. If $0<\kappa_{\cN}<1$, \cref{thm:integrated-envelope} gives $\Venv=\Omega(S\log S)=\Omega(T\log T)$. If $\kappa_{\cN}=1$, the same theorem excludes error below $1/2$ after one nonempty time step.
\end{proof}

\section{Erasure-noise lower bound}
\label{sec:supp-erasure-lower}

{In \cref{lem:adaptive-survival,thm:erasure-one,prop:path-dimension,cor:erasure-additive}, we prove, respectively,} the survival bound, the one-qubit lower bound, the dimension bound, and their combined $K$-qubit form used in the \hyperref[sec:main-erasure]{main-text erasure section} and the \hyperref[sec:methods-erasure-lower]{corresponding Methods subsection}.

\subsection{Complete erasure}

Fix $p\in(0,1)$. At every time step, each physical qubit is independently erased with probability $p$, and the erased locations are known to the recovery operation. The width at $\partial_s$ is selected from the earlier measurement record $h_{s-1}$, before the erasure pattern at time step $s$ is known. We assume that every complete measurement record satisfies $\sum_s q_s(h_{s-1})\leq C$.

\begin{lemma}[Complete erasure removes entanglement with the reference]\label{lem:catastrophe}
Feed half of a Bell state into the memory and retain the other half in a reference system $R$. Condition on the first time step in which every physical qubit is erased, including the time-step index and the preceding measurement record. Thereafter, the joint state of $R$ and every system available to the recovery operation is separable. The conditional final Bell-state fidelity is at most $1/2$.
\end{lemma}

\begin{proof}
The erasure channel replaces every physical qubit by an orthogonal erasure state that is independent of its input. After complete erasure, the recovery operation retains only erasure locations, classical records, and newly prepared states that were initially independent of $R$. The conditional state is separable across the reference and recovery systems. Later operations act only on the recovery systems and cannot create entanglement with $R$. Apply \cref{lem:bell-witness} to the final two-qubit state.
\end{proof}

\begin{lemma}[Probability of avoiding complete erasure]\label{lem:adaptive-survival}
Let $P_{\mathrm{surv}}(S,C)$ be the largest probability, over all adaptive protocols whose every execution has circuit size at most $C$, that no complete-erasure event occurs. Then
\begin{equation}
P_{\mathrm{surv}}(S,C)
\leq\max_{\substack{n_1,\ldots,n_S\in\mathbb Z_{\geq0}\\\sum_sn_s\leq C}}
\prod_{s=1}^{S}(1-p^{n_s})
\leq e^{-S p^{C/S}}.
\label{eq:supp-survival}
\end{equation}
\end{lemma}

\begin{proof}
Let $F(s,v)$ be the largest survival probability when $s$ time steps remain and at most $v$ further storage locations may be used. %
Clearly $F(0,v)=1$. A deterministic next width $n\leq v$ survives the next time step with probability $1-p^n$; conditioned on any surviving erasure pattern, the remaining success probability is at most $F(s-1,v-n)$. A randomized choice of the next width is an average over deterministic choices and therefore cannot exceed the largest deterministic value. %
Thus
\begin{equation}
F(s,v)\leq\max_{0\leq n\leq v}(1-p^n)F(s-1,v-n).
\end{equation}
Induction on $s$ proves the first inequality of \cref{eq:supp-survival}.

For a fixed feasible allocation, $1-x\leq e^{-x}$ gives
\begin{equation}
\prod_s(1-p^{n_s})\leq\exp\!\left[-\sum_s p^{n_s}\right].
\end{equation}
The function $p^x$ is convex and decreasing. Jensen and $\sum_sn_s\leq C$ imply
\begin{equation}
\sum_s p^{n_s}\geq S p^{(\sum_sn_s)/S}\geq S p^{C/S}.
\end{equation}
Substituting this lower bound into the preceding exponential upper bound proves the second inequality in \cref{eq:supp-survival}.
\end{proof}

\begin{theorem}[One-qubit circuit-size bound for erasure noise]\label{thm:erasure-one}
Every one-qubit memory lasting $S$ time steps for which each individual execution has circuit size at most $C$ obeys
\begin{equation}
1-F_{\mathrm e}\geq\frac12(1-e^{-S p^{C/S}}).
\label{eq:supp-erasure-fe}
\end{equation}
If its diamond error is at most $\eps<1$, let $c_\eps=-\log(1-\eps)$. Whenever $S>c_\eps$,
\begin{equation}
C\geq\frac{S}{\log(1/p)}\log\frac{S}{c_\eps}.
\label{eq:supp-erasure-volume}
\end{equation}
\end{theorem}

\begin{proof}
By \cref{lem:adaptive-survival}, complete erasure occurs with probability at least $1-e^{-S p^{C/S}}$. Partition the records according to the first complete-erasure event. On every such record, \cref{lem:catastrophe} bounds the fidelity by $1/2$; when complete erasure does not occur, the fidelity is at most one. Averaging gives \cref{eq:supp-erasure-fe}.

Diamond error $\eps$ implies $1-F_{\mathrm e}\leq\eps/2$. Combining with \cref{eq:supp-erasure-fe} yields
\begin{equation}
e^{-S p^{C/S}}\geq1-\eps,\qquad S p^{C/S}\leq c_\eps.
\end{equation}
Taking logarithms and using $\log p=-\log(1/p)$ gives \cref{eq:supp-erasure-volume}.
\end{proof}

\subsection{Dimension along one execution}

\begin{proposition}[One execution retains the full logical dimension]\label{prop:path-dimension}
Let an adaptive protocol with a finite or countable set of classical measurement records store $K\geq1$ logical qubits for $S$ time steps. Let $\cM_K$ denote the decoded channel on the $K$-qubit logical input. If $\|\cM_K-\id_K\|_\diamond\leq\eps<1$, then some complete record with nonzero probability satisfies
\begin{equation}
q_s(h_{s-1})\geq K\quad\text{for all }s,
\end{equation}
and therefore $\Vpath\geq KS$.
\end{proposition}

\begin{proof}
Let $D=2^K$ and input half of $\Phi_D$, retaining the reference $R$. Suppose instead that every complete record with nonzero probability contains at least one time step with $q_s\leq K-1$. Conditioned on such a record, the physical register at that time step has dimension at most $2^{K-1}$, so its joint state with $R$ has Schmidt number at most $2^{K-1}$. {Each later branch is obtained by a local trace-nonincreasing map on the memory system, so conditioning and normalization do not increase the Schmidt number.}  Hence every normalized conditional output state, and therefore their mixture, has Schmidt number at most $2^{K-1}$. By \cref{lem:singlet-schmidt}, its overlap with $\Phi_D$ is at most $1/2$, whereas diamond-norm closeness gives overlap at least $1-\eps/2>1/2$. This contradiction proves the claim.
\end{proof}

\begin{corollary}[Combined circuit-size lower bound for erasure noise]\label{cor:erasure-additive}
If a $K$-qubit erasure memory lasting $S$ time steps has diamond error at most $\eps<1$, then
\begin{equation}
\Vpath\geq KS
\end{equation}
and, for $S>c_\eps$,
\begin{equation}
\Vpath\geq\frac{S}{\log(1/p)}\log\frac{S}{c_\eps}.
\end{equation}
For every fixed $\eps_0<1$, uniformly over $0<\eps\leq\eps_0$ and sufficiently large $S$,
\begin{equation}
\Vpath=\Omega\!\left(S\left(K+\log\frac{S}{\eps}\right)\right).
\label{eq:supp-erasure-additive}
\end{equation}
\end{corollary}

\begin{proof}
The dimension term follows from \cref{prop:path-dimension}. Fixing $K-1$ logical inputs and discarding their outputs produces a one-qubit memory with no larger worst-case circuit size, so \cref{thm:erasure-one} gives the second term. For $0<\eps\leq\eps_0<1$, $c_\eps=-\log(1-\eps)$ lies between positive constant multiples of $\eps$. Hence $\log(S/c_\eps)={\Theta(\log(S/\eps))}$. Combining the two lower bounds proves the result.
\end{proof}

\section{CSS-code upper bound for erasure noise}
\label{sec:supp-css-upper}

In this section we give detailed results on a memory protocol that uses a family of CSS codes under erasure noise.
\cref{thm:css-gv} gives positive-rate CSS codes with linear distance, \cref{lem:chernoff} bounds the probability that one time step contains too many erasures, and \cref{thm:wide-erasure} combines these facts into a fixed-width memory construction. Together with the erasure lower bound, \cref{cor:tight-erasure} proves the scaling stated in \cref{thm:tight-erasure,eq:main-relative}. 

\subsection{Code existence and concentration}

Define the binary entropy $h_2(x)=-x\log_2x-(1-x)\log_2(1-x)$, and let $h_2^{-1}:[0,1]\to[0,1/2]$ denote the inverse of its restriction to $[0,1/2]$. Set $\delta_{\mathrm{GV}}=h_2^{-1}(1/2)\simeq0.1100$.

\begin{theorem}[Binary CSS Gilbert--Varshamov bound]\label{thm:css-gv}
    For any real $\Delta,R$ such that $0<\Delta<\delta_{\mathrm{GV}}$ and $0<R<1-2h_2(\Delta)$, there exists $N_0$ such that for every integer $N\geq N_0$ a binary CSS stabilizer code exists with parameters $[[N,K_N,d_N]]$ satisfying
    \begin{equation}\label{eqn:css-gv}
        K_N\geq RN,\qquad d_N\geq\Delta N.
    \end{equation}
\end{theorem}

\begin{proof}
    Since $R<1-2h_2(\Delta)$ and $h_2$ is continuous, we may choose $\Delta'$ such that
    \begin{equation}
        \Delta<\Delta'<\delta_{\mathrm{GV}},
        \qquad
        R<1-2h_2(\Delta').
    \end{equation}
    We may then choose $R'$ such that
    \begin{equation}
        R<R'<1-2h_2(\Delta').
        \label{eq:supp-intermediate-css-parameters}
    \end{equation}

    In the notation of Ref.~\cite{Ashikhmin2014}, an $(N,k_{1,N},k_{2,N})$ CSS code is defined by binary linear codes $C_1\subseteq C_2$ of dimensions $k_{1,N}$ and $k_{2,N}$. Now let
    \begin{equation}
        r_1=\frac{1-R'}2,
        \qquad
        r_2=1-r_1=\frac{1+R'}2,
    \end{equation}
    and, for each $N$, choose
    \begin{equation}
        k_{1,N}=\lfloor r_1N\rfloor,
        \qquad
        k_{2,N}=N-k_{1,N}.
        \label{eq:supp-balanced-css-parameters}
    \end{equation}
    This choice is balanced in the sense of Ref.~\cite{Ashikhmin2014}, namely $k_{1,N}=N-k_{2,N}$. 
    Equation~(29) of~\cite{Ashikhmin2014} gives a binary generator matrix of rank $k_{1,N}+N-k_{2,N}$ for the associated self-orthogonal code. 
    The corresponding CSS stabilizer code therefore encodes
    \begin{equation}
        K_N
        =N-(k_{1,N}+N-k_{2,N})
        =k_{2,N}-k_{1,N}
        =N-2k_{1,N}.
    \end{equation}
    Since $k_{1,N}=\lfloor r_1N\rfloor\leq r_1N$, it follows that
    \begin{equation}
        K_N\geq N-2r_1N=R'N>RN.
        \label{eq:supp-balanced-css-rate}
    \end{equation}

    Let $d_{\mathrm{GV}}(N,k_{1,N},k_{2,N})$ be the largest integer defined in Theorem~19 of Ref.~\cite{Ashikhmin2014}. That theorem states that an $(N,k_{1,N},k_{2,N})$ CSS code exists with minimum distance
    \begin{equation}
        d_N=d_{\mathrm{GV}}(N,k_{1,N},k_{2,N}).
        \label{eq:supp-finite-css-distance}
    \end{equation}
    Moreover, $k_{1,N}/N\to r_1$ and $k_{2,N}/N\to r_2$. 
    Applying Eq.~(42) of Ref.~\cite{Ashikhmin2014} to this sequence of integer dimensions gives
    \begin{equation}
    \begin{aligned}
        \lim_{N\to\infty}\frac{d_N}{N}
        &=h_2^{-1}\!\left(\min\{r_1,1-r_2\}\right)
        =h_2^{-1}\!\left(r_1\right)
        =h_2^{-1}\!\left(\frac{1-R'}2\right).
        \label{eq:supp-balanced-css-distance}
    \end{aligned}
    \end{equation}
    Here the $O(1)$ rounding in \cref{eq:supp-balanced-css-parameters} changes the constituent-code rates by only $O(1/N)$ and hence does not change the limit. 
    The balanced condition equates the two arguments in the minimum in Eq.~(42) of Ref.~\cite{Ashikhmin2014}, making $r_1=1-r_2$; it does not assert that the realized $X$- and $Z$-distances are equal.

    Because $h_2$ is strictly increasing on $(0,1/2)$, the inequality $R'<1-2h_2(\Delta')$ in \cref{eq:supp-intermediate-css-parameters} is equivalent to
    \begin{equation}
        \Delta'<h_2^{-1}\!\left(\frac{1-R'}2\right).
    \end{equation}
    It follows from \cref{eq:supp-balanced-css-distance} that there exists $N_0$ such that, for every $N\geq N_0$,
    \begin{equation}
        \frac{d_N}{N}\geq\Delta'>\Delta.
    \end{equation}
    Together with \cref{eq:supp-balanced-css-rate}, this gives
    \begin{equation}
        K_N\geq RN,
        \qquad
        d_N\geq\Delta N
    \end{equation}
    for every $N\geq N_0$, proving \cref{eqn:css-gv}.
\end{proof}

\begin{remark}
    The condition $R<1-2h_2(\Delta)$ places the target rate strictly below the balanced CSS Gilbert--Varshamov guaranteed-rate curve at relative distance $\Delta$. 
    The original asymptotic form of this tradeoff appears in Ref.~\cite{CalderbankShor1996}. 
    The proof above uses the finite-length CSS bound and its asymptotic evaluation from Ref.~\cite{Ashikhmin2014} in order to obtain a code for every sufficiently large block length.
\end{remark}

\begin{lemma}[Binomial Chernoff bound]\label{lem:chernoff}
    For a binomial random variable $X\sim\mathrm{Bin}(N,p)$ (i.e. $\Pr[X=x]=\binom{N}{x}p^x(1-p)^{N-x}$ for $x\in\{0,\dots,N\}$) and real numbers $p,\Delta$ such that $p<\Delta<1$, it holds that
    \begin{equation}
        \Pr[X\geq\Delta N]\leq e^{-N D(\Delta\|p)},
    \end{equation}
    where $D(\Delta\|p)$ is the binary Kullback--Leibler divergence with natural logarithms, given by
    \begin{equation}
        D(\Delta\|p)=\Delta\log\frac\Delta p+(1-\Delta)\log\frac{1-\Delta}{1-p}>0.
    \end{equation}
\end{lemma}

\begin{proof}
    Write $X=\sum_{i=1}^N B_i$, where the $B_i$ are independent Bernoulli random variables with success probability $p$. Fix $\theta>0$ and apply Markov's inequality to the nonnegative random variable $Y=e^{\theta X}$ with threshold $a=e^{\theta\Delta N}$. 
    Since the function $x\mapsto e^{\theta x}$ is increasing,
    \begin{equation}
    \begin{aligned}
        \Pr[X\geq\Delta N]
        &=\Pr[e^{\theta X}\geq e^{\theta\Delta N}]
        \leq e^{-\theta\Delta N}\mathbb E[e^{\theta X}].
        \label{eq:supp-chernoff-markov}
    \end{aligned}
    \end{equation}
    Independence among $\{B_i\}_{i=1}^N$ gives
    \begin{equation}
        \mathbb E [e^{\theta X}]
        =\prod_{i=1}^N\mathbb E [e^{\theta B_i}]
        =(1-p+pe^\theta)^N.
        \label{eq:supp-chernoff-mgf}
    \end{equation}
    Combining \cref{eq:supp-chernoff-markov,eq:supp-chernoff-mgf} yields
    \begin{equation}
        \Pr[X\geq\Delta N]
        \leq\left(e^{-\theta\Delta}(1-p+pe^\theta)\right)^N.
    \end{equation}
    The logarithm of the factor on the right-hand side is
    \begin{equation}
        f(\theta)=-\theta\Delta+\log(1-p+pe^\theta)
        \quad\text{with}\quad
        \frac{d}{d\theta}f(\theta)=-\Delta+\frac{pe^\theta}{1-p+pe^\theta}.
    \end{equation}
    Its unique minimizer $\theta_*$ over $\theta>0$ satisfies
    \begin{equation}
        e^{\theta_*}=\frac{\Delta(1-p)}{p(1-\Delta)},
    \end{equation}
    which is larger than one because $\Delta>p$. 
    Substituting $\theta=\theta_*$ gives
    \begin{equation}
        e^{-\theta_*\Delta}(1-p+pe^{\theta_*})
        =\left(\frac p\Delta\right)^\Delta
        \left(\frac{1-p}{1-\Delta}\right)^{1-\Delta}
        =e^{-D(\Delta\|p)}.
    \end{equation}
    Hence $\Pr[X\geq\Delta N]\leq e^{-ND(\Delta\|p)}$. Finally, $D(\Delta\|p)>0$ because $\Delta\neq p$.
\end{proof}

\subsection{Block size and logical error}

\begin{theorem}[CSS-code construction for erasure noise]\label{thm:wide-erasure}
    Consider real numbers $p,\Delta,R$ such that $0<p<\Delta<\delta_{\mathrm{GV}}$ and $0<R<1-2h_2(\Delta)$, and set $a=D(\Delta\|p)$.
    For all integers $K,S\geq1$ and $0<\eps\leq1$, the $K$ logical qubits can be encoded into one $[[N,K_N,d_N]]$ CSS code block with $K_N\geq K$ to give a memory with ideal recovery, diamond error at most $\eps$, and fixed width
    \begin{equation}
        N\leq c_0\left(K+1+\log\frac{2S}{\eps}\right),
        \label{eq:supp-wide-width}
    \end{equation}
    where $c_0$ depends only on $p,\Delta,R$. 
    Since the width is fixed, both $\Venv$ and $\Vpath$ are equal to $NS$.
    Here $K$ is the number of logical input qubits, $K_N$ is the number of logical qubits encoded by the selected $N$-qubit CSS code block, and $S$ is the number of memory time steps, equivalently, the number of applications of the erasure-noise channel.
\end{theorem}

\begin{proof}
    Let $N_0$ be as in \cref{thm:css-gv} for a fixed pair $(\Delta,R)$ satisfying $0<\Delta<\delta_{\mathrm{GV}}$ and $0<R<1-2h_2(\Delta)$.
    Choose the least integer $N$ satisfying
    \begin{equation}
    N\geq\max\!\left\{N_0,\frac{K+1}{R},\frac1a\log\frac{2S}{\eps}\right\}.
    \label{eq:supp-choose-n}
    \end{equation}
    By \cref{thm:css-gv}, the inequality $N\geq N_0$ ensures that this value of $N$ admits an $[[N,K_N,d_N]]$ CSS code with $K_N\geq RN$ and $d_N\geq\Delta N$.
    Moreover, $N\geq(K+1)/R$ and $R>0$ imply $RN\geq K+1$, and hence $K_N\geq K+1>K$.
    Encode the $K$-qubit input into any chosen $K$ of the $K_N$ logical qubits of this code and fix the remaining $K_N-K$ logical inputs to known pure states $|\psi\>$:
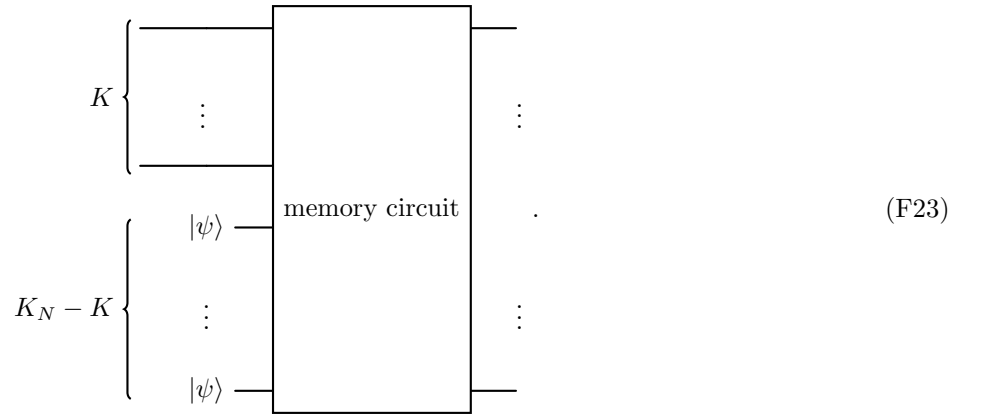
    \begin{equation}
        \begin{quantikz}
            \lstick[3]{$K$} && \gate[6]{\text{memory circuit}} & \\
            & \vdots\; \setwiretype{n} && \;\vdots\\
            &&& \setwiretype{n} \\
            \lstick[3]{$K_N-K$} & \invgate{|\psi\>} \setwiretype{n} & \setwiretype{q} & \setwiretype{n} \\
            & \vdots \setwiretype{n} && \;\vdots \\
            & \invgate{|\psi\>} \setwiretype{n} & \setwiretype{q} &
        \end{quantikz}.
    \end{equation}
    The diagram is schematic at the logical level: the encoder maps these $K_N$ logical inputs into the $N$ physical qubits of the code block.
    
    At one time step, let $[N]=\{1,\ldots,N\}$ and identify an erasure pattern with its erasure set $E\in2^{[N]}$, where $2^{[N]}$ is the power set of $[N]$. 
    Under independent erasure noise, $\Pr[E]=p^{|E|}(1-p)^{N-|E|}$, and the number $X=|E|$ of erased physical qubits satisfies $X\sim\mathrm{Bin}(N,p)$.
    We encode the input in the distance-$d_N$ CSS code chosen above.
    Call the erasure set $E$ good when $|E|<d_N$ and bad when $|E|\geq d_N$, and define $P_{\mathrm{bad}}=\Pr[|E|\geq d_N]=\Pr[X\geq d_N]$.
    Every good erasure set is exactly correctable because the erased locations are known. 
    Since $d_N\geq\Delta N$, the event $\{X\geq d_N\}$ is contained in $\{X\geq\Delta N\}$. 
    Applying \cref{lem:chernoff} with Bernoulli parameter $p$ and threshold $\Delta>p$, and using $a=D(\Delta\|p)$, gives
    \begin{equation}
        P_{\mathrm{bad}}\leq\Pr[X\geq\Delta N]\leq e^{-aN}.
    \end{equation}
    
    For an erasure set $E$, let $\Lambda_E$ be the recovered logical channel on the $K$ encoded logical input qubits conditioned on $E$, with its output re-encoded in the same block and the other $K_N-K$ logical qubits returned to their fixed states before the next time step, and let $\Lambda=\sum_{E\in2^{[N]}}\Pr(E)\Lambda_E$ be the channel averaged over the random erasure set $E$. 
    For good erasure sets, we have $\Lambda_E=\id$; for bad erasure sets, $\|\Lambda_E-\id\|_\diamond\leq2$. 
    Hence
    \begin{equation}\label{eq:noisy_memory_channel_diamond_norm}
        \|\Lambda-\id\|_\diamond
        \leq \sum_{\substack{E\in2^{[N]}\\ |E|\geq d_N}} \Pr(E)\|\Lambda_E-\id\|_\diamond
        \leq 2P_{\mathrm{bad}}
        \leq 2e^{-aN}.
    \end{equation}
    
    Because the erasure sets at different time steps are independent and the same code, erasure channel, and recovery are used at every time step, write 
    \begin{equation}
        \Lambda^S = \underbrace{\Lambda\circ\cdots\circ\Lambda}_{S\text{ times}}
    \end{equation}
    for the recovered logical channel after $S$ time steps.
    The telescoping bound $\|\Lambda_T\cdots\Lambda_1-
    \Gamma_T\cdots\Gamma_1\|_\diamond
    \leq\sum_{t=1}^{T}
    \|\Lambda_t-\Gamma_t\|_\diamond$ in \cref{eq:supp-telescope} gives, upon taking $T=S$, $\Lambda_t=\Lambda$, and $\Gamma_t=\id$,
    \begin{equation}
        \|\Lambda^S-\id\|_\diamond
        \leq S \|\Lambda-\id\|_\diamond 
        \leq 2Se^{-aN}\leq\eps.
    \end{equation}
    The final inequality follows from $N\geq a^{-1}\log(2S/\eps)$ (c.f.~\cref{eq:supp-choose-n}).
    Here $N$ is the least integer satisfying \cref{eq:supp-choose-n}.
    Writing the maximum on the right-hand side of \cref{eq:supp-choose-n} as $M$, the least integer $N$ satisfying $N\geq M$ is $\lceil M\rceil$, so
    \begin{equation}
        N=\lceil M\rceil\leq M+1
        \leq 1+N_0+\frac{K+1}{R}+\frac1a\log\frac{2S}{\eps}
        \leq c_0\left(K+1+\log\frac{2S}{\eps}\right),
    \end{equation}
    for a constant $c_0$ depending only on $N_0,R,$ and $a$; for example, one may take $c_0=1+N_0+R^{-1}+a^{-1}$. 
    Since $N_0$ depends only on $\Delta,R$ and $a=D(\Delta\|p)$, this constant depends only on $p,\Delta,R$. 
    This proves \cref{eq:supp-wide-width}.

    Lastly, let $\mathcal Q$ denote this fixed-width memory protocol. 
    Since the construction uses one $N$-qubit code block at every time step, $q_s(h_{s-1})=N$ for every $s$ and every preceding measurement record $h_{s-1}$. 
    Hence every execution has circuit size
    \begin{equation}
        C(\mathcal Q;h_S)=\sum_{s=1}^{S}q_s(h_{s-1})=NS,
    \end{equation}
    and it follows that
    \begin{equation}
        \Vpath(\mathcal Q)
        =\sup_{h_S}\sum_{s=1}^{S}q_s(h_{s-1})
        =NS
        =\sum_{s=1}^{S}\sup_{h_{s-1}}q_s(h_{s-1})
        =\Venv,
    \end{equation}
    concluding the proof.
\end{proof}

\begin{corollary}[Optimal worst-case circuit size]\label{cor:tight-erasure}
    Fix real numbers $p,\eps_0$ such that $0<p<\delta_{\mathrm{GV}}$ and $0<\eps_0<1$. 
    For adaptive protocols under independent erasure noise and ideal recovery,
    \begin{equation}
        C_{\min}(K,S,\eps)
        =\Theta\!\left(
        S\left(K+\log\frac{S}{\eps}\right)
        \right).
    \end{equation}
    More precisely, there exist constants $c_-,c_+>0$ and an integer $S_0\geq1$, depending only on $p$ and $\eps_0$, such that, for all integers $K\geq1$ and $S\geq S_0$ and every $\eps\in(0,\eps_0]$,
    \begin{equation}
        c_-S\left(K+\log\frac{S}{\eps}\right)
        \leq C_{\min}(K,S,\eps)
        \leq c_+S\left(K+\log\frac{S}{\eps}\right).
    \end{equation}
\end{corollary}

\begin{proof}
    The lower bound is \cref{cor:erasure-additive}, valid for every adaptive protocol.
    Taking the infimum over admissible protocols gives the same lower bound for $C_{\min}(K,S,\eps)$, with constants and a lower threshold for $S$ depending only on $p$ and $\eps_0$.
    By \cref{thm:wide-erasure}, the fixed-width memory protocol construction $\mathcal Q$ using a CSS code from \cref{thm:css-gv} has
    \begin{equation}
        C_{\min}(K,S,\eps)\leq\Vpath(\mathcal Q)=NS
        \leq
        c_0S\left(K+1+\log\frac{2S}{\eps}\right)
        =
        O\!\left(
        S\left(K+\log\frac{S}{\eps}\right)
        \right).
    \end{equation}
    After fixing $p$, choose
    \begin{equation}
        \Delta=\frac{p+\delta_{\mathrm{GV}}}{2},
        \qquad
        R=\frac{1-2h_2(\Delta)}{2}.
    \end{equation}
    Then $\Delta,R,a$, and $c_0$ depend only on $p$. For $K,S\geq1$ and $0<\eps\leq\eps_0$, we have $K+1\leq2K$ and
    \begin{equation}
        \log\frac{2S}{\eps}
        \leq
        \left(
        1+\frac{\log2}{\log(1/\eps_0)}
        \right)
        \log\frac{S}{\eps}.
    \end{equation}
    Thus the constant in the upper bound depends only on $p$ and $\eps_0$. Combining it with the uniform lower bound proves the stated inequalities.
\end{proof}

The dimension and reliability terms are comparable when $K=\Theta(\log(S/\eps))$.
The construction is an existence proof with ideal recovery; it does not provide an efficient decoder or a geometrically local implementation.

\section{Circuit-level upper bound}
\label{sec:supp-circuit-upper}

\begin{definition}[Logical gadget]\label{def:logical-gadget}
    A \emph{logical gadget} is a physical circuit built from preparations, elementary gates, measurements, waits, and classical feedforward. In the absence of noise, it implements one ideal logical layer on encoded data exactly.
    Its output contains every quantum system and classical measurement outcome passed to the next logical layer.
    Let $\widetilde{\mathcal G}_h$ and $\mathcal G_h$ denote, respectively, the noisy and ideal logical quantum--classical channels conditioned on a permitted preceding measurement record $h$. 
    An error bound $\delta$ is \emph{composable} if, for every $h$, every external reference system $R$, and every joint logical-reference input state $\rho_{LR}$,
    \begin{equation}
        \left\|
        \left[
        \left(\widetilde{\mathcal G}_h-\mathcal G_h\right)
        \otimes\id_R
        \right](\rho_{LR})
        \right\|_1
        \leq\delta.
    \end{equation}
    Equivalently,
    \begin{equation}
        \left\|\widetilde{\mathcal G}_h-\mathcal G_h\right\|_\diamond
        \leq\delta
    \end{equation}
    uniformly over permitted preceding records. Because these are quantum--classical channels, the comparison includes both the quantum output and the distribution of the newly produced classical outcomes.
\end{definition}

\begin{assumption}[Positive-rate gadget family with exponentially small error]\label{ass:gadgets}
    Fix a physical noise model. 
    There are constants $R_{\mathrm{gad}},c_q,c_t,\beta>0$, $A\geq1$, $c_N\geq1$, and $N_0\geq1$, and a code family with a specified set $\mathcal N_{\mathrm{gad}}\subseteq\mathbb N$ of available block lengths, satisfying the following properties.
    When $N\in\mathcal N_{\mathrm{gad}}$, the corresponding code block and all the gadgets specified below exist under the fixed physical noise model.
    \begin{enumerate}[label=(\roman*)]
        \item A length-$N$ code block encodes at least $R_{\mathrm{gad}}N$ logical qubits.
        \item For every real $x\geq N_0$, there exists an available block length $N\in\mathcal N_{\mathrm{gad}}$ such that
        \begin{equation}
            x\leq N\leq c_Nx.
        \end{equation}
        \item For every logical layer composed of operations from a specified gate set, including the state preparations and measurements in that set, that acts on at most $R_{\mathrm{gad}}N$ logical qubits, there is a corresponding gadget. 
        The set need only contain the operations used by the ideal circuit; universality is required only if the statement is to cover arbitrary circuits.
        \item The gadget has live physical width at most $c_qN$ at every time step, including every physical quantum system used by the gadget.
        \item It lasts at most $c_t$ time steps, including the classical computation and feedback time used for decoding and recovery.
        \item For each logical layer $t$, available block length $N$, and permitted preceding measurement record $h$, let $\widetilde{\mathcal G}_{N,t,h}$ be the effective logical quantum--classical channel implemented by the noisy gadget, including its newly produced classical outcomes, and let $\mathcal G_{t,h}$ be the corresponding ideal logical-layer channel. For every external reference system $R$ that is not acted on by the gadget and every joint logical-reference input state $\rho_{LR}$,
        \begin{equation}
            \left\|
            \left(
            \left(
            \widetilde{\mathcal G}_{N,t,h}
            -\mathcal G_{t,h}
            \right)\otimes\id_R
            \right)(\rho_{LR})
            \right\|_1
            \leq Ae^{-\beta N}.
            \label{eq:supp-gadget-branch}
        \end{equation}
        Equivalently,
        \begin{equation}
            \left\|
            \widetilde{\mathcal G}_{N,t,h}
            -\mathcal G_{t,h}
            \right\|_\diamond
            \leq Ae^{-\beta N},
        \end{equation}
        uniformly over $t$ and $h$.
        Here $A\geq1$ is a fixed finite-block-length prefactor and $\beta>0$ is the block-length error-suppression exponent. 
        They may depend on the fixed physical noise model and the code-and-gadget family, but are independent of $N$, the logical layer $t$, the preceding record $h$, the logical input, and the external reference system. 
        Thus $\delta_N:=Ae^{-\beta N}$ is a uniform upper bound on the composable error of one logical gadget.
        \item Every quantum system and newly produced classical outcome passed to the next gadget is included in the gadget output. No uncounted quantum memory or environment is carried between gadgets.
    \end{enumerate}
\end{assumption}

\begin{remark}
    This assumption is stronger than currently established general-purpose compiler theorems with constant space overhead~\cite{YamasakiKoashi2024,Tamiya2026,NguyenPattison2025}. 
    In particular, it requires a complete logical layer, rather than an isolated logical gate, to use $O(N)$ physical qubits for $O(1)$ time steps with 
    error at most $Ae^{-\beta N}$, uniformly over the logical layer, preceding record, logical input, and external reference system. 
    These are sufficient assumptions for \cref{thm:circuit-achievability}; we do not assert that a current architecture satisfies them.
\end{remark}

\begin{theorem}[Conditional circuit achievability]\label{thm:circuit-achievability}
    Under \cref{ass:gadgets}, consider an ideal circuit composed of the specified logical layers, of width $K\geq1$ and depth $T\geq1$, and let $\mathcal C$ denote its overall ideal quantum--classical channel. For every target diamond-norm error $0<\eps\leq1$, there exist an available block length $N\in\mathcal N_{\mathrm{gad}}$ and a corresponding simulation whose overall noisy quantum--classical channel $\widetilde{\mathcal C}_N$ and implementation depth $S$ satisfy
    \begin{equation}
        \left\|
        \widetilde{\mathcal C}_N-\mathcal C
        \right\|_\diamond
        \leq\eps,
        \qquad
        S\leq c_tT.
        \label{eq:supp-circuit-achievability-accuracy-depth}
    \end{equation}
    The simulation has worst-case physical circuit size $\CFT$ satisfying
    \begin{equation}
        \CFT
        \leq
        \cgad T\left(
        K+1+\log\frac{AT}{\eps}
        \right),
        \label{eq:supp-circuit-achievability-size}
    \end{equation}
    where $\cgad$ depends only on the gadget family. Here $\CFT$ is the supremum, over realizable complete classical measurement records, of the total number of physical preparation, gate, measurement, and wait locations in the corresponding execution. Therefore
    \begin{equation}
        \frac{\CFT}{KT}
        =
        O\!\left(
        1+\frac{\log(T/\eps)}K
        \right).
    \end{equation}
\end{theorem}

\begin{proof}
    \Cref{ass:gadgets}(i)--(iii) provide an available code block with sufficient logical capacity and a gadget for every ideal logical layer; \cref{ass:gadgets}(vi) controls the simulation error; and \cref{ass:gadgets}(iv), (v), and (vii) control the implementation depth and physical location count. The lower and upper inequalities in \cref{ass:gadgets}(ii) are used separately below.
    
    Now, let us set
    \begin{equation}\label{eq:block_length_multiplier}
        x=\max\!\left\{N_0,\frac{K}{R_{\mathrm{gad}}},
        \frac1\beta\log\frac{AT}{\eps}\right\}.
    \end{equation}
    By \cref{ass:gadgets}(ii), choose an available block length $N\in[x,c_Nx]$.
    The inequality $N\geq x$, together with the definition of $x$, gives
    \begin{equation}
        N\geq\frac{K}{R_{\mathrm{gad}}},
        \qquad\text{and hence}\qquad
        R_{\mathrm{gad}}N\geq K.
        \label{eq:supp-gadget-capacity}
    \end{equation}
    By \cref{ass:gadgets}(i), the length-$N$ code block therefore encodes at least $R_{\mathrm{gad}}N\geq K$ logical qubits and can accommodate the entire width-$K$ logical register. Moreover, every layer of a width-$K$ ideal circuit acts on at most
    \begin{equation}
        K\leq R_{\mathrm{gad}}N
    \end{equation}
    logical qubits. Since the circuit is composed of operations from the specified gate set, \cref{ass:gadgets}(iii) therefore supplies a corresponding physical gadget for every one of its $T$ logical layers.

    Set $\delta_N=Ae^{-\beta N}$.
    For $t=1,\ldots,T$, let $\mathsf H_{t-1}$ denote the set of measurement records preceding logical layer $t$.
    By \cref{ass:gadgets}(vi), for every layer $t=1,\ldots,T$ and every permitted preceding record $h\in \mathsf H_{t-1}$,
    \begin{equation}\label{eq:supp-gadget-branch2}
        \left\|
        \widetilde{\mathcal G}_{N,t,h}
        -\mathcal G_{t,h}
        \right\|_\diamond
        \leq \delta_N.
    \end{equation}
    Then, let $\widetilde{\mathfrak G}_{N,t}$ and $\mathfrak G_t$ be the quantum--classical channels that, on the history block labelled by $h\in\mathsf H_{t-1}$, act respectively as $\widetilde{\mathcal G}_{N,t,h}$ and $\mathcal G_{t,h}$ and retain the updated classical record in their outputs.
    To make the averaging over preceding measurement records explicit, let $L$ denote the logical data system and consider an arbitrary quantum--classical input, possibly entangled with an external reference $R$, of the form
    \begin{equation}
        \rho_{LHR}
        =
        \sum_{h\in\mathsf H_{t-1}}
        p_h\,\rho^h_{LR}\otimes
        |h\rangle\!\langle h|_H.
    \end{equation}
    Here $p_h$ is the probability of the preceding record $h$ generated by the part of the relevant hybrid circuit before layer $t$. 
    So for all $\rho$ over $L$, we can write
    \begin{equation}
        \widetilde{\mathfrak G}_{N,t}(\rho\otimes\ketbra{h}) = \widetilde{\mathcal G}_{N,t,h}(\rho),
        \qquad
        \mathfrak G_t(\rho\otimes\ketbra{h}) = \mathcal G_{t,h}(\rho).
    \end{equation}
    Thus, by the triangle inequality for the trace norm and \cref{eq:supp-gadget-branch2}, we have
    \begin{align}
        \left\|
        \left(
        \left(
        \widetilde{\mathfrak G}_{N,t}
        -\mathfrak G_t
        \right)\otimes\id_R
        \right)(\rho_{LHR})
        \right\|_1
        &\leq
        \sum_{h\in\mathsf H_{t-1}}
        p_h
        \left\|
        \left(\left(
        \widetilde{\mathcal G}_{N,t,h}
        -\mathcal G_{t,h}
        \right)\otimes\id_R \right)
        (\rho^h_{LR})
        \right\|_1
        \notag\\
        &\leq
        \sum_{h\in\mathsf H_{t-1}}p_h\delta_N
        =\delta_N.
    \end{align}
    Because the branchwise estimate is uniform in $h$, no fixed or common history distribution needs to be assumed, and
    \begin{equation}
        \left\|
        \widetilde{\mathfrak G}_{N,t}
        -\mathfrak G_t
        \right\|_\diamond
        \leq\delta_N.
    \end{equation}

    Write the noisy and ideal overall quantum--classical channels as
    \begin{equation}
        \widetilde{\mathcal C}_N
        =
        \widetilde{\mathfrak G}_{N,T}
        \circ\cdots\circ
        \widetilde{\mathfrak G}_{N,1},
        \qquad
        \mathcal C
        =
        \mathfrak G_T\circ\cdots\circ\mathfrak G_1,
    \end{equation}
    respectively.
    The telescoping identity in \cref{eq:supp-telescope} gives
    \begin{align}
        \widetilde{\mathcal C}_N-\mathcal C
        =
        \sum_{t=1}^{T}
        &\widetilde{\mathfrak G}_{N,T}
        \circ\cdots\circ \widetilde{\mathfrak G}_{N,t+1}
        \circ \left( \widetilde{\mathfrak G}_{N,t} - \mathfrak G_t \right) 
        \circ \mathfrak G_{t-1} \circ\cdots\circ \mathfrak G_1.
    \end{align}
    Each summand is the change produced by replacing the $t$-th noisy gadget by its ideal controlled logical layer in the corresponding hybrid circuit. 
    Submultiplicativity of the diamond norm and the unit diamond norm of every channel therefore give
    \begin{equation}
        \left\| \widetilde{\mathcal C}_N-\mathcal C \right\|_\diamond
        \leq \sum_{t=1}^{T} \left\| \widetilde{\mathfrak G}_{N,t} -\mathfrak G_t \right\|_\diamond
        \leq T\delta_N = TAe^{-\beta N}
    \end{equation}
    Now, the inequality $N\geq x$ supplied by \cref{ass:gadgets}(ii) also gives
    \begin{equation}
        N\geq\frac1\beta\log\frac{AT}{\eps},
    \end{equation}
    which implies that,
    \begin{equation}
        TAe^{-\beta N}
        \leq
        TA\exp\!\left(-\log\frac{AT}{\eps}\right)
        =\eps.
    \end{equation}
    Therefore the simulated circuit satisfies
    \begin{align}\label{eq:diamond_norm_error_bound}
        \left\|\widetilde{\mathcal C}_N-\mathcal C\right\|_\diamond
        \leq TAe^{-\beta N}
        \leq\eps.
    \end{align}
    This establishes that the chosen simulation meets the target diamond-norm error of $\eps$ and hence is an admissible simulation for \cref{thm:circuit-achievability}. 

    We next make the location count explicit.
    Fix a realizable complete classical record $h_T$. 
    For the gadget implementing logical layer $t$, let $s_t(h_T)$ be its duration and let $q_{t,u}(h_T)$ be its live physical width at its $u$th time step. By \cref{ass:gadgets}(v),
    \begin{equation}
        s_t(h_T)\leq c_t
        \qquad
        \text{for every }t,
    \end{equation}
    and hence the implementation depth along this execution satisfies
    \begin{equation}
        S(h_T)
        =\sum_{t=1}^{T}s_t(h_T)
        \leq c_tT.
    \end{equation}
    Taking the supremum over complete measurement records gives
    \begin{equation}
        S:=\sup_{h_T}S(h_T)\leq c_tT.
        \label{eq:supp-gadget-depth}
    \end{equation}
    
    By \cref{ass:gadgets}(iv),
    \begin{equation}
        q_{t,u}(h_T)\leq c_qN
    \end{equation}
    at every physical time step of every gadget. In the synchronized circuit model, the number of physical preparation, gate, measurement, and wait locations at that time step is at most the live physical width. 
    Therefore, for every realizable complete measurement record $h_T$,
    \begin{align}
        C_{\mathrm{phys}}(\widetilde{\mathcal C}_N;h_T)
        &\leq
        \sum_{t=1}^{T}
        \sum_{u=1}^{s_t(h_T)}
        q_{t,u}(h_T)
        \leq
        \sum_{t=1}^{T}
        \sum_{u=1}^{s_t(h_T)}
        c_qN
        =
        c_qN\sum_{t=1}^{T}s_t(h_T)
        =
        c_qN S(h_T).
    \end{align}
    Here \cref{ass:gadgets}(vii) ensures that every quantum system retained between successive gadgets is included in the gadget width and hence in this location count. Taking the supremum over complete measurement records and using \cref{eq:supp-gadget-depth} gives
    \begin{align}
        \CFT
        &=
        \sup_{h_T}
        C_{\mathrm{phys}}(\widetilde{\mathcal C}_N;h_T)
        \leq
        c_qN\sup_{h_T}S(h_T)
        =
        c_qNS
        \leq
        c_qc_tNT.
        \label{eq:supp-gadget-circuit-size}
    \end{align}

    Now we use the inequality $N\leq c_Nx$ in \cref{ass:gadgets}(ii) and \cref{eq:block_length_multiplier} to obtain
    \begin{align}
        N
        &\leq
        c_N
        \max\!\left\{
        N_0,\frac{K}{R_{\mathrm{gad}}},
        \frac1\beta\log\frac{AT}{\eps}
        \right\}
        \leq
        c_N\left(
        N_0+\frac{K}{R_{\mathrm{gad}}}
        +\frac1\beta\log\frac{AT}{\eps}
        \right)
        \leq
        b_{\mathrm{gad}}
        \left(
        K+1+\log\frac{AT}{\eps}
        \right),
        \label{eq:supp-gadget-block-length}
    \end{align}
    where
    \begin{equation}
        b_{\mathrm{gad}}
        =
        c_N\max\!\left\{
        N_0,R_{\mathrm{gad}}^{-1},\beta^{-1}
        \right\}.
    \end{equation}
    Substituting the block-length bound in \cref{eq:supp-gadget-block-length} into the location-count bound in \cref{eq:supp-gadget-circuit-size} gives
    \begin{align}\label{eq:location_count_upper_bound}
        \CFT
        &\leq c_qc_tNT
        \leq
        c_qc_tb_{\mathrm{gad}}T
        \left(
        K+1+\log\frac{AT}{\eps}
        \right)
        =
        \cgad T
        \left(
        K+1+\log\frac{AT}{\eps}
        \right),
    \end{align}
    where $\cgad=c_qc_tb_{\mathrm{gad}}$.
    
    Thus the target-error requirement determines the logarithmic contribution to $N$, while the width and depth bounds convert the selected block length into the asserted worst-case physical circuit size.
    Therefore, $\cgad$ depends only on the constants appearing in \cref{ass:gadgets}.
    Dividing by $KT$ and using $\log(AT/\eps)=\log A+\log(T/\eps)$ gives
    \begin{align}
        \frac{\CFT}{KT}
        &\leq
        \cgad\left( 1+\frac1K+\frac{\log A}{K} +\frac{\log(T/\eps)}K \right)
        =
        O\!\left(
        1+\frac{\log(T/\eps)}K
        \right),
    \end{align}
    since $A$ is fixed and $K\geq1$. 
    
    The diamond-norm estimate in \cref{eq:diamond_norm_error_bound} and \cref{eq:supp-gadget-depth} prove the accuracy and depth claims in \cref{eq:supp-circuit-achievability-accuracy-depth}, respectively, while the preceding location-count estimate in \cref{eq:location_count_upper_bound} proves \cref{eq:supp-circuit-achievability-size}.
\end{proof}

We do not claim that any existing architecture satisfies \cref{ass:gadgets} or attains the circuit-level scaling in \cref{thm:circuit-achievability}.
We give sufficient code-and-gadget assumptions under which physical circuit size has the same additive dependence on logical width and target error. 
A direct comparison with the memory converse also requires an explicit relation between the ideal circuit depth $T$ and the implementation depth $S$.

\section{Subsystem spacetime-code bounds}
\label{sec:supp-spacetime}

In \cref{ass:coordinate-volume,thm:st-singleton,thm:pauli-pair,cor:accuracy-distance}, we state the circuit-to-code assumption and prove the bounds used in the main text and \hyperref[sec:methods-spacetime]{the corresponding Methods subsection}; in \cref{rem:pairing-example}, we give the three-qubit repetition-code example.

\subsection{Circuit-to-code relation}

Fix a fault-tolerant circuit and a selected list of $\Nst$ elementary Pauli fault locations. Suppose fault propagation and measured checks define an exact binary linear subsystem code
\begin{equation}
    [[\Nst,k,r,\dst]],
\end{equation}
with stabilizer $\mathcal S_{\mathrm{st}}$, gauge group $\mathcal G_{\mathrm{st}}$ and
\begin{equation}
    \dst=\min\{\wt(L):L\in N(\mathcal S_{\mathrm{st}})\setminus\mathcal G_{\mathrm{st}}\}.
\end{equation}
Such correspondences arise naturally for Clifford circuits and Pauli reductions of broader protocols~\cite{Bacon2017,DelfossePaetznick2023,Pesah2025}. 
We do not assert such a correspondence for every fault-tolerant protocol.

The code length counts selected Pauli fault locations, whereas physical circuit size counts storage locations. We therefore impose the following assumption separately for each circuit-size measure to which it is applied.

\begin{assumption}\label{ass:coordinate-volume}
    For the specified circuit-to-code location count $\Cmap$, there is a scale-independent constant $\cloc>0$ such that
    \begin{equation}
        \Nst\leq\cloc\Cmap.
        \label{eq:supp-coordinate-volume}
    \end{equation}
\end{assumption}

The assumption holds when each physical circuit location is associated with at most a constant number of selected fault coordinates. It may fail when the spacetime-code coordinates include the union of mutually exclusive adaptive branches. In addition, a Pauli string in the abstract code must be realizable by the circuit fault model, with the stated weight and probability, before a code-level correction statement can be transferred to the circuit.

For example, consider a protocol with two time steps that uses one physical qubit at time step 1, produces one of $M$ classical outcomes, and uses one outcome-dependent physical qubit at time step 2. Every physical execution has
\begin{equation}
\Vpath=\Venv=2.
\end{equation}
A coordinate list formed from the union of all outcome-dependent locations instead has $\Nst=1+M$. The comparison $\Nst\leq\cloc\Vpath$ would require $\cloc\geq(M+1)/2$, which is not independent of the problem size. This example does not assert that the union defines an exact subsystem code; it demonstrates why the coordinate-to-circuit-size comparison must be stated explicitly.

\subsection{Singleton and exact correction}

\begin{theorem}[Subsystem Singleton circuit-size bound]\label{thm:st-singleton}
    Every exact binary linear subsystem spacetime code $[[\Nst,k,r,\dst]]$ with $k\geq1$ satisfies
    \begin{equation}
        \Nst\geq k+r+2(\dst-1).
        \label{eq:supp-st-singleton}
    \end{equation}
    Under \cref{ass:coordinate-volume},
    \begin{equation}
        \Cmap\geq\frac{k+r+2(\dst-1)}{\cloc}.
        \label{eq:supp-st-volume}
    \end{equation}
\end{theorem}

\begin{proof}
    If $\dst=1$, the code space contains a $k$-qubit protected subsystem and an $r$-qubit gauge subsystem, so it has dimension $2^{k+r}$ inside an $\Nst$-qubit space and therefore $\Nst\geq k+r$.  
    If $\dst\geq2$, the subsystem Singleton inequality
    \cref{eq:supp-singleton-base} reads
    \begin{equation}
        k+r
        \leq
        \Nst-2(\dst-1).
    \end{equation}
    Adding $2(\dst-1)$ to both sides gives
    \begin{equation}
        k+r+2(\dst-1)
        \leq
        \Nst,
    \end{equation}
    which is \cref{eq:supp-st-singleton}.
    
    By \cref{eq:supp-coordinate-volume},
    \begin{equation}
        \Nst\leq\cloc\Cmap.
    \end{equation}
    Since $\cloc>0$, division by $\cloc$ preserves the direction of the inequality and gives
    \begin{equation}
        \Cmap\geq\frac{\Nst}{\cloc}.
    \end{equation}
    Combining this inequality with \cref{eq:supp-st-singleton} yields
    \begin{equation}
        \Cmap
        \geq
        \frac{\Nst}{\cloc}
        \geq
        \frac{k+r+2(\dst-1)}{\cloc},
    \end{equation}
    which proves \cref{eq:supp-st-volume}.
\end{proof}

\begin{corollary}[Cost of correcting every small Pauli fault]\label{cor:small-pauli}
    Assume the relevant Pauli strings are realized by circuit faults and \cref{ass:coordinate-volume} holds. If every Pauli fault of weight at most $t$ is corrected exactly, then
    \begin{equation}
        \dst\geq2t+1,\qquad
        \Cmap\geq\frac{k+r+4t}{\cloc}.
    \end{equation}
\end{corollary}

\begin{proof}
    If $\dst\leq2t$, choose a dressed logical $L$ of weight at most $2t$. 
    Partition its support into two sets of size at most $t$ and write $L=E_1E_2$. Because $L$ has zero syndrome, $E_1$ and $E_2$ have the same syndrome. 
    No recovery can correct both: otherwise their quotient $L$ would act trivially on the protected subsystem, contrary to its definition. 
    Thus $\dst\geq2t+1$. 
    The inequality $\dst\geq2t+1$ implies
    \begin{equation}
        \dst-1\geq2t,
        \qquad\text{and hence}\qquad
        2(\dst-1)\geq4t.
    \end{equation}
    Therefore \cref{eq:supp-st-volume} gives
    \begin{align}
        \Cmap
        &\geq
        \frac{k+r+2(\dst-1)}{\cloc}
        \geq
        \frac{k+r+4t}{\cloc},
    \end{align}
    as claimed.
\end{proof}

\subsection{Same-syndrome fault bound}

Let $F$ be a phase-free Pauli fault string with distribution $\nu$. Its complete syndrome $s(F)$ consists of commutation signs with a fixed generating set of $\mathcal S_{\mathrm{st}}$. The decoder %
may use internal randomness $U$ chosen independently of $F$ and chooses a Pauli recovery $R(s(F),U)$ having the same syndrome.
The residual is
\begin{equation}\label{eq:residual_error}
    L_{\mathrm{res}}(F,U)=R(s(F),U)F\in N(\mathcal S_{\mathrm{st}}).
\end{equation}
Failure occurs when $L_{\mathrm{res}}\in N(\mathcal S_{\mathrm{st}})\setminus\mathcal G_{\mathrm{st}}$. The only fault-dependent input to the decoder is the syndrome; no erasure flag, analogue likelihood, or leakage information distinguishes faults with the same syndrome.

For a dressed logical $L$, 
define the $L$-shifted distribution $\nu_L(F)=\nu(FL)$ 
and overlap
\begin{equation}
    \Ov_\nu(L)=\sum_F\min\{\nu(F),\nu(FL)\}
    =1-\|\nu-\nu_L\|_{\mathrm{TV}}.
\end{equation}
Here
$\|\nu-\nu_L\|_{\mathrm{TV}}=\frac12\sum_F|\nu(F)-\nu_L(F)|$ is the total-variation distance.
\begin{theorem}[Pauli-pairing bound]\label{thm:pauli-pair}
    Fix an exact binary linear subsystem code
    $[[N_{\rm st},k,r,d_{\rm st}]]$ with $k\geq1$.
    For every deterministic or randomized complete-syndrome decoder defined above and every dressed logical Pauli $L\in N(\cS_{\rm st})\setminus\cG_{\rm st}$,
        \begin{equation}
         P_{\rm L}
         \geq
         \frac12\operatorname{Ov}_{\nu}(L)
         \label{eq:pairing-overlap}
    \end{equation}
    In addition, suppose that $\nu$ is a full-support product Pauli distribution: there are single-coordinate probability distributions $\pi_i$ on $\{I,X,Y,Z\}$ such that
    \begin{equation}
    \begin{aligned}
        \nu(F)
        &=
        \prod_{i=1}^{\Nst}\pi_i(F_i)
        &&\text{for every }
        F\in\{I,X,Y,Z\}^{\Nst},
        \\
        \pi_i(P)
        &>0
        &&\text{for every }
        1\leq i\leq\Nst
        \text{ and }
        P\in\{I,X,Y,Z\},
        \\
        \sum_{P\in\{I,X,Y,Z\}}\pi_i(P)
        &=1
        &&\text{for every }
        1\leq i\leq\Nst.
    \end{aligned}
    \end{equation}
    Then set
    \begin{equation}
    \begin{gathered}
        q_{\min}
        :=
        \min_{\substack{1\leq i\leq N_{\rm st}\\P\in\{I,X,Y,Z\}}}
        \pi_i(P),
        \qquad
        q_{\max}
        :=
        \max_{\substack{1\leq i\leq N_{\rm st}\\P\in\{I,X,Y,Z\}}}
        \pi_i(P),
        \qquad
        \rho:=\frac{q_{\min}}{q_{\max}}.
    \end{gathered}
    \end{equation}
    Then it holds that
    \begin{equation}
     P_{\rm L}
     \geq
     \frac{\rho^{d_{\rm st}}}
          {1+\rho^{d_{\rm st}}}.
     \label{eq:pairing}
    \end{equation}
    Full support implies that every phase-free Pauli string has strictly positive probability and that $0<q_{\min}\leq q_{\max}$. 
    Hence $\rho\in(0,1]$, and the likelihood ratios used above are well-defined.
\end{theorem}

\begin{proof}
    Fix a dressed logical Pauli 
    (which may act on the gauge subsystem but acts nontrivially on the protected subsystem). %
    $L\in\mathrm N(\cS_{\rm st})\setminus\cG_{\rm st}$.
    Multiplication by $L$ partitions the phase-free Pauli faults into disjoint unordered pairs. 
    Let $\mathcal P_{\Nst}$ denote the set of phase-free Pauli strings on the $\Nst$ selected fault coordinates, and let $\mathcal F_L$ denote the resulting set of unordered pairs.
    \begin{equation}
        \mathcal F_L=\{\{F,FL\}:F\in\mathcal P_{\Nst}\}.
    \end{equation}
    The two members have the same syndrome because $L$ commutes with every element of $\cS_{\rm st}$.

    For a fixed decoder seed $U$, the decoder uses the same recovery $R$ on $F$ and $FL$.
    The two residual operators are $RF$ and $RFL=(RF)L$.
    They cannot both belong to $\cG_{\rm st}$, since otherwise their quotient would imply $L\in\cG_{\rm st}$.
    Thus either $F$ or $FL$ causes a logical failure, and the failure contribution of that pair is at least $\min\{\nu(F),\nu(FL)\}$.  Summing over fault pairs $\{F,FL\}$ gives
    \begin{equation}
    \begin{aligned}
        P_{\rm L}
         &\ge \sum_{\{F,FL\}\in {\mathcal F_L}} \min\{\nu(F),\nu(FL)\}
         =\frac12\sum_F\min\{\nu(F),\nu(FL)\},
    \end{aligned}
    \end{equation}
    which proves \cref{eq:pairing-overlap}.  Since this holds for every fixed seed, averaging proves it for a randomized decoder.

    Now assume full-support product Pauli noise and choose $L$ with $\wt(L)=d_{\rm st}$.
    Since $\nu(F)=\prod_i\pi_i(F_i)$ and multiplication by $L$ changes only the coordinates in $\supp(L)$,
    \begin{equation}
    \frac{\nu(FL)}{\nu(F)}
    =\prod_{i\in\supp(L)}
    \frac{\pi_i((FL)_i)}{\pi_i(F_i)}.
    \end{equation}
    Every factor lies in $[\rho,\rho^{-1}]$, where $\rho=q_{\min}/q_{\max}$. Therefore
    \begin{equation}
    \rho^{d_{\rm st}}
    \leq\frac{\nu(FL)}{\nu(F)}
    \leq\rho^{-d_{\rm st}}.
    \label{eq:pair-ratio}
    \end{equation}
    For one pair $\{F,FL\}$, 
    let $m_F=\min\{\nu(F),\nu(FL)\}$ and $M_F=\max\{\nu(F),\nu(FL)\}$. \cref{eq:pair-ratio} gives $m_F/M_F\geq\rho^{d_{\rm st}}$, and hence
    \begin{equation}
    m_F\geq
    \frac{\rho^{d_{\rm st}}}{1+\rho^{d_{\rm st}}}(m_F+M_F).
    \end{equation}
    Summing over the disjoint pairs, whose total probability is one, yields
    \begin{equation}
    P_{\rm L}\geq
    \frac{\rho^{d_{\rm st}}}{1+\rho^{d_{\rm st}}},
    \end{equation}
    which proves \cref{eq:pairing}.
\end{proof}

\begin{remark}[Three-qubit repetition code]\label{rem:pairing-example}
    For the bit-flip repetition code, $S_1=Z_1Z_2$, $S_2=Z_2Z_3$, and $\overline X=X_1X_2X_3$. The same-syndrome pairs are
    \begin{equation}
    I\leftrightarrow X_1X_2X_3,\quad
    X_1\leftrightarrow X_2X_3,\quad
    X_2\leftrightarrow X_1X_3,\quad
    X_3\leftrightarrow X_1X_2.
    \end{equation}
    Under independent $X$ faults with probability $p\leq1/2$, majority decoding has logical failure probability $3p^2(1-p)+p^3=3p^2-2p^3=\tfrac12\Ov_\nu(\overline X)$ and therefore attains the overlap bound. This code does not correct arbitrary single-qubit Pauli faults.
\end{remark}

\begin{remark}[Identity probability matters]
    The extrema $\qmin,\qmax$ include all four Pauli symbols. For symmetric depolarizing noise with total nonidentity probability $p<3/4$,
    \begin{equation}
    \begin{aligned}
    \pi(I)=1-p,\qquad
    \pi(X)=\pi(Y)=\pi(Z)=p/3,\qquad
    \rho=\frac{p}{3(1-p)}.
    \end{aligned}
    \end{equation}
    The case $\qmin=\qmax$ forces all four probabilities to equal $1/4$; it is the completely depolarizing Pauli distribution, not a rare-error limit.
\end{remark}

\begin{corollary}[Accuracy forces distance]\label{cor:accuracy-distance}
    Under the full-support product-noise hypotheses of \cref{thm:pauli-pair}, if $\qmin<\qmax$ and $\PL\leq\eps<1/2$, then
    \begin{equation}
        \dst\geq\frac{\log[(1-\eps)/\eps]}{\log(\qmax/\qmin)}
        \label{eq:supp-distance-bound}
    \end{equation}
    and
    \begin{equation}
        \Nst\geq\max\!\left\{k+r,\ k+r-2+
        \frac{2\log[(1-\eps)/\eps]}{\log(\qmax/\qmin)}\right\}.
        \label{eq:supp-length-bound}
    \end{equation}
    Under \cref{ass:coordinate-volume}, it holds that
    \begin{equation}
        \Cmap
        \geq
        \frac{1}{\cloc}
        \max\!\left\{
        k+r,\,
        k+r-2+
        \frac{2\log[(1-\eps)/\eps]}
             {\log(\qmax/\qmin)}
        \right\}.
        \label{eq:supp-accuracy-volume-bound}
    \end{equation}
    If $\qmin=\qmax$, then $\PL\geq1/2$.
\end{corollary}

\begin{proof}
    Suppose first that $\qmin<\qmax$. 
    Full-support product Pauli noise then implies
    \begin{equation}
        0<\rho=\frac{\qmin}{\qmax}<1.
    \end{equation}
    Combining \cref{eq:pairing} of \cref{thm:pauli-pair} with $\PL\leq\eps$ gives
    \begin{equation}
        \frac{\rho^{\dst}}{1+\rho^{\dst}}
        \leq
        \PL
        \leq
        \eps.
    \end{equation}
    Since $1+\rho^{\dst}>0$ and $1-\eps>0$, it follows that
    \begin{align}
        \rho^{\dst}
        &\leq
        \eps\bigl(1+\rho^{\dst}\bigr)
        \notag\\
        \Longrightarrow\qquad
        (1-\eps)\rho^{\dst}
        &\leq
        \eps
        \notag\\
        \Longrightarrow\qquad
        \rho^{\dst}
        &\leq
        \frac{\eps}{1-\eps}.
    \end{align}
    Taking the logarithm of both sides of the last inequality gives
    \begin{equation}
        \dst\log\rho
        \leq
        \log\frac{\eps}{1-\eps}.
    \end{equation}
    Because $\log\rho<0$, division by $\log\rho$ reverses the direction of the inequality. Therefore
    \begin{align}
        \dst
        &\geq
        \frac{\log[\eps/(1-\eps)]}{\log\rho}
        =
        \frac{\log[(1-\eps)/\eps]}
             {\log(\qmax/\qmin)},
    \end{align}
    which proves \cref{eq:supp-distance-bound}.
    
    Substituting \cref{eq:supp-distance-bound} into
    \cref{eq:supp-st-singleton} gives
    \begin{align}
        \Nst
        &\geq
        k+r+2(\dst-1)
        \geq
        k+r-2+
        \frac{2\log[(1-\eps)/\eps]}
             {\log(\qmax/\qmin)}.
        \label{eq:supp-length-bound-reliability}
    \end{align}
    Since $\dst\geq1$,
    \cref{eq:supp-st-singleton} also gives
    \begin{equation}
        \Nst
        \geq
        k+r+2(\dst-1)
        \geq
        k+r.
        \label{eq:supp-length-bound-dimension}
    \end{equation}
    Taking the stronger bound of
    \cref{eq:supp-length-bound-reliability,eq:supp-length-bound-dimension}
    proves \cref{eq:supp-length-bound}.

    Under \cref{ass:coordinate-volume},
    \cref{eq:supp-coordinate-volume} and $\cloc>0$ imply
    \begin{equation}
        \Cmap\geq\frac{\Nst}{\cloc}.
    \end{equation}
    Combining this inequality with \cref{eq:supp-length-bound} proves
    \cref{eq:supp-accuracy-volume-bound}.

    Finally, if $\qmin=\qmax$, then $\rho=1$, and
    \cref{eq:pairing} gives
    \begin{equation}
        \PL
        \geq
        \frac{1^{\dst}}{1+1^{\dst}}
        =
        \frac12,
    \end{equation}
    as claimed.
\end{proof}

\begin{corollary}[Vanishing error forces growing protection]\label{cor:vanishing-distance}
    For every scale $\lambda$, let $[[\Nst(\lambda),k(\lambda),r(\lambda),\dst(\lambda)]]$ be an exact binary linear subsystem spacetime code with $k(\lambda)\geq1$, equipped with a deterministic or randomized complete-syndrome decoder and a full-support product Pauli fault distribution as in \cref{thm:pauli-pair}. 
    Let
    $\qmin(\lambda)$ and $\qmax(\lambda)$ be the corresponding single-coordinate probability extrema and define
    \begin{equation}
        \rho_\lambda
        :=
        \frac{\qmin(\lambda)}{\qmax(\lambda)}.
    \end{equation}
    Assume that there is a scale-independent constant
    $\rho_0\in(0,1]$ such that
    \begin{equation}
        \rho_\lambda\geq\rho_0
        \qquad
        \text{for every scale }\lambda.
    \end{equation}
    If $\PL(\lambda)\to0$, then
    \begin{equation}
        \dst(\lambda)\to\infty,\qquad
        \Nst(\lambda)-k(\lambda)-r(\lambda)\to\infty.
    \end{equation}
\end{corollary}

\begin{proof}
    Suppose, for contradiction, that $\dst(\lambda)$ does not tend to
    infinity. Then there are a finite constant $D$ and an infinite
    subsequence along which $\dst(\lambda)\leq D$. For every $\lambda$
    in this subsequence,
    \begin{equation}
        \rho_\lambda^{\dst(\lambda)}
        \geq
        \rho_\lambda^D
        \geq
        \rho_0^D,
    \end{equation}
    where the first inequality uses
    $0<\rho_\lambda\leq1$ and $\dst(\lambda)\leq D$, and the second uses
    $\rho_\lambda\geq\rho_0$. Since the function
    $x\mapsto x/(1+x)$ is increasing for $x>0$, \cref{eq:pairing} gives
    \begin{equation}
        \PL(\lambda)\geq
        \frac{\rho_\lambda^{\dst(\lambda)}}
             {1+\rho_\lambda^{\dst(\lambda)}}
        \geq
        \frac{\rho_0^D}{1+\rho_0^D}>0.
    \end{equation}
    The final lower bound is independent of $\lambda$, contradicting
    $\PL(\lambda)\to0$. Hence $\dst(\lambda)\to\infty$.

    Applying the subsystem code Singleton bound \cref{eq:supp-st-singleton} at every scale and subtracting
    $k(\lambda)+r(\lambda)$ from both sides gives
    \begin{equation}
        \Nst(\lambda)-k(\lambda)-r(\lambda)
        \geq
        2\bigl(\dst(\lambda)-1\bigr).
    \end{equation}
    Since the right side tends to infinity due to $\dst(\lambda)\to\infty$ established above, so does the left side.
    This proves the asserted divergence of the absolute spacetime-code
    redundancy.
\end{proof}

\noindent
The Singleton bound controls the spacetime-code length; it gives a physical circuit-size bound only under \cref{ass:coordinate-volume}. The overlap inequality holds for arbitrary Pauli fault distributions under syndrome-only Pauli recovery, whereas the distance-dependent form uses the product-noise assumption. The results require growing absolute redundancy as the logical error tends to zero, but they do not require growing relative overhead when the number of protected logical qubits increases simultaneously. %

\section{Distance and threshold conditions}
\label{sec:supp-threshold}

Distance growth is necessary under the product-Pauli assumptions of \cref{thm:pauli-pair}, but distance alone does not imply a threshold. 
A code family may have increasing distance while the number of fault patterns that cause logical failure grows too rapidly. 
We state a sufficient condition for the spacetime-code family to have a positive decoder threshold by counting malignant fault sets.

Let $\lambda$ index a spacetime-code family and set
\begin{equation}\label{eq:supp-correctable-radius}
    t_\lambda=\left\lfloor\frac{\dst(\lambda)-1}{2}\right\rfloor.
\end{equation}
At scale $\lambda$, let $F$ be an induced phase-free Pauli fault string and let $U$ denote the decoder randomness, with $U$ omitted for a deterministic decoder. 
A joint realization $(F,U)$ is a \emph{failed fault pattern} if the residual defined in \cref{eq:residual_error} of the preceding section is a nontrivial dressed logical Pauli:
\begin{equation}
    L_{\mathrm{res}}(F,U)
    \in
    N(\mathcal S_{\mathrm{st}})
    \setminus
    \mathcal G_{\mathrm{st}}.
\end{equation}
Its faulty-coordinate support is
\begin{equation}
    \supp(F)
    :=
    \{i:F_i\neq I\}.
\end{equation}
A set $W$ of fault coordinates is called a malignant fault set if there are nonidentity Pauli labels on the coordinates in $W$ and, when the decoder is randomized, a value of $U$, for which a fault pattern with support exactly $W$ is failed.

Let $\mathcal W_{\lambda,w}$ be a chosen collection of malignant fault sets of cardinality $w$ over spacetime code with index $\lambda$. 
These collections \emph{cover decoder failure} if, for every failed realization $(F,U)$, there are some $w\geq t_\lambda+1$ and some
$W\in\mathcal W_{\lambda,w}$ such that
\begin{equation}
    W\subseteq\supp(F).
\end{equation}
In particular, the restriction $w\geq t_\lambda+1$ is part of the covering hypothesis and requires that no fault pattern of weight at most $t_\lambda$ is failed.
If Pauli labels are counted separately, they may be included in the constant $\mu$ below.
Let $p\geq0$ denote the underlying physical noise-strength parameter, with $p=0$ corresponding to the noiseless physical circuit. 
Let $F_{\lambda,p}$ denote the induced Pauli fault string at scale $\lambda$, and let $\PL(\lambda,p)$ denote the probability, over both $F_{\lambda,p}$ and the independent decoder randomness, of a failed realization. 
We say that the family has a \emph{positive decoder threshold} precisely when $p_{\mathrm{th}}>0$ for the quantity defined in \cref{eq:supp-decoder-threshold}.

\begin{theorem}[Sufficient condition for a positive threshold]\label{thm:witness-threshold}
    Assume:
    \begin{enumerate}[label=(\roman*)]
        \item $\dst(\lambda)\to\infty$;
        \item For every physical noise strength $p\geq0$, the induced fault strings are local stochastic with a scale-uniform rate $p_{\mathrm{st}}(p)\in[0,1]$: for every scale $\lambda$ and every set
        $A\subseteq\{1,\ldots,\Nst(\lambda)\}$,
        \begin{equation}
            \Pr_p\!\left[
            A\subseteq\supp(F_{\lambda,p})
            \right]
            \leq
            p_{\mathrm{st}}(p)^{|A|}.
            \label{eq:supp-induced-local-stochastic}
        \end{equation}
        Moreover,
        \begin{equation}
            p_{\mathrm{st}}(0)=0,
            \qquad
            p_{\mathrm{st}}(p)\longrightarrow0
            \quad\text{as }p\to0.
        \end{equation}
        \item the chosen decoder corrects every induced Pauli fault pattern of weight at most $t_\lambda$, the malignant sets cover decoder failure, and for constants
        $C,\mu>0$ independent of $\lambda$ and $p$,
        \begin{equation}
            |\mathcal W_{\lambda,w}|
            \leq
            C\Nst(\lambda)\mu^w;
        \end{equation}
        \item $\log\Nst(\lambda)=o(\dst(\lambda))$.
    \end{enumerate}
    If $\mu p_{\mathrm{st}}(p)<1$, then
    \begin{equation}
    \PL(\lambda,p)\leq
    \frac{C\Nst(\lambda)[\mu p_{\mathrm{st}}(p)]^{t_\lambda+1}}
    {1-\mu p_{\mathrm{st}}(p)}\longrightarrow0.
    \label{eq:supp-threshold}
    \end{equation}
    Consequently, there exists $p_0>0$ such that
    \begin{equation}
        \lim_{\lambda\to\infty}\PL(\lambda,p)=0
        \qquad
        \text{for every fixed }0\leq p<p_0.
    \end{equation}
    Thus $p_{\mathrm{th}}\geq p_0>0$ in the sense of
    \cref{eq:supp-decoder-threshold}.
\end{theorem}

\begin{proof}
    For each $W\in\mathcal W_{\lambda,w}$, local stochasticity bounds the probability that all coordinates in $W$ are faulty, irrespective of faults outside $W$, by $p_{\mathrm{st}}(p)^w$.
    Let $\zeta=\mu p_{\mathrm{st}}(p)$.
    By the $t_\lambda$-error-correction hypothesis, no fault pattern of weight at most $t_\lambda$ causes decoder failure. By the covering hypothesis, every remaining failure contains some
    $W\in\mathcal W_{\lambda,w}$ with $w\geq t_\lambda+1$. Therefore the failure event is contained in
    \begin{equation}
        \bigcup_{w=t_\lambda+1}^{\Nst(\lambda)}
        \ \bigcup_{W\in\mathcal W_{\lambda,w}}
        \left\{
        W\subseteq\supp(F_{\lambda,p})
        \right\},
    \end{equation}
    which explains the lower limit $t_\lambda+1$ in the following union bound:
    \begin{equation}
        \{\text{decoder failure}\}
        \subseteq
        \bigcup_{w=t_\lambda+1}^{\Nst(\lambda)}
        \ \bigcup_{W\in\mathcal W_{\lambda,w}}
        \left\{
        W\subseteq\supp(F_{\lambda,p})
        \right\}.
        \label{eq:supp-failure-witness-cover}
    \end{equation}
    Applying the union bound to \cref{eq:supp-failure-witness-cover} and then using
    \cref{eq:supp-induced-local-stochastic} gives
    \begin{equation}
        \PL(\lambda,p)
        \leq\sum_{w=t_\lambda+1}^{\Nst}|\mathcal W_{\lambda,w}|p_{\mathrm{st}}(p)^w
        \leq C\Nst\sum_{w=t_\lambda+1}^{\Nst}\zeta^w.
    \end{equation}

    When $0\leq\zeta<1$, every term is nonnegative, so adding the terms with
    $w>\Nst(\lambda)$ can only increase the sum:
    \begin{align}
        \sum_{w=t_\lambda+1}^{\Nst(\lambda)}\zeta^w
        &\leq
        \sum_{w=t_\lambda+1}^{\infty}\zeta^w
        =
        \zeta^{t_\lambda+1}
        \sum_{j=0}^{\infty}\zeta^j
        =
        \frac{\zeta^{t_\lambda+1}}{1-\zeta}.
    \end{align}
    Consequently,
    \begin{align}
        \PL(\lambda,p)
        &\leq
        C\Nst(\lambda)
        \sum_{w=t_\lambda+1}^{\Nst(\lambda)}\zeta^w
        \leq
        \frac{
        C\Nst(\lambda)\zeta^{t_\lambda+1}
        }{1-\zeta}
        =
        \frac{
        C\Nst(\lambda)
        [\mu p_{\mathrm{st}}(p)]^{t_\lambda+1}
        }{
        1-\mu p_{\mathrm{st}}(p)
        },
    \end{align}
    where the final equality uses
    $\zeta=\mu p_{\mathrm{st}}(p)$. This is
    \cref{eq:supp-threshold}.
    If $\zeta=0$, the conclusion is immediate. 
    If $0<\zeta<1$, the logarithm of the numerator of \cref{eq:supp-threshold}  is, apart from the additive constant $\log C$,
    \begin{equation}
        \log\Nst+(t_\lambda+1)\log\zeta.
    \end{equation}
    Moreover, \cref{eq:supp-correctable-radius} gives
    \begin{equation}
        t_\lambda+1
        =
        \left\lfloor
        \frac{\dst(\lambda)-1}{2}
        \right\rfloor+1
        \geq
        \frac{\dst(\lambda)}{2}.
    \end{equation}
    Since $0<\zeta<1$, one has $\log\zeta<0$, and therefore
    \begin{align}
        \log\Nst(\lambda)
        +(t_\lambda+1)\log\zeta
        &\leq
        \log\Nst(\lambda)
        +\frac{\dst(\lambda)}2\log\zeta
        =
        o(\dst(\lambda))
        +\frac{\dst(\lambda)}2\log\zeta
        \longrightarrow-\infty.
    \end{align}
    Here we used assumptions (i) and (iv). 
    Hence the numerator in
    \cref{eq:supp-threshold} tends to zero, while
    $1-\zeta>0$ is independent of $\lambda$.
    Thus we have proved that
    $\PL(\lambda,p)\to0$ for every fixed $p$ satisfying
    $\mu p_{\mathrm{st}}(p)<1$.
    Note that the factor $\zeta^{t_\lambda+1}$ provides exponential suppression in the spacetime-code distance, while assumption (iv) ensures that the prefactor $\Nst(\lambda)$ grows too slowly to overcome this suppression.
    
    Finally, since $p_{\mathrm{st}}(p)\to0$ as $p\to0$ by assumption (ii) and $\mu>0$ is independent of $p$, there exists $p_0>0$ such that
    \begin{equation}
        p_{\mathrm{st}}(p)<\frac1\mu
        \qquad
        \text{for every }0\leq p<p_0.
    \end{equation}
    Thus $\mu p_{\mathrm{st}}(p)<1$ throughout the interval $[0,p_0)$, and the preceding argument gives
    \begin{equation}
        \lim_{\lambda\to\infty}\PL(\lambda,p)=0
        \qquad
        \text{for every fixed }0\leq p<p_0.
    \end{equation}
    By \cref{eq:supp-decoder-threshold},
    $p_{\mathrm{th}}\geq p_0>0$, which proves the positive-threshold claim.
\end{proof}

\noindent {The same-syndrome fault bound and \cref{cor:vanishing-distance} are converse results. The theorem above is a sufficient condition for a positive threshold: it controls the number of malignant fault sets as a function of their weight. It does not imply optimal physical circuit size.}

\section{Consequences for resource scaling}
\label{sec:supp-resource-scaling}

\subsection{Resource lower bounds from error converses}

\cref{thm:pauli-pair,cor:accuracy-distance} combine a lower bound on logical error with a size--distance relation. We state this conversion in general form in the following proposition.

\begin{proposition}[Accuracy-dependent cost from a converse]\label{prop:generic-converse}
    Fix a physical noise strength $p$.
    Suppose a family of codes or fault-tolerant gadgets is indexed by its distance $d\in\mathbb N$, and let $C_{\mathrm{res}}(d)$ denote one specified scalar resource count for its distance-$d$ member. 
    The choice of resource is fixed throughout the proposition. Depending on the application, $C_{\mathrm{res}}(d)$ may denote the physical-qubit count, the spacetime-code length $\Nst(d)$, the circuit-to-code location count $\Cmap(d)$, or the worst-case physical circuit size $\CFT(d)$. 
    No equivalence among these resource counts is assumed.
    
    Consider
    \begin{equation}
        a(p)>0,\qquad
        b(p)>0,\qquad
        c_{\mathrm{res}}>0,\qquad
        \beta>0,\qquad
        \gamma>0,
    \end{equation}
    independent of $d$, such that
    \begin{equation}
        \PL(d,p)\geq a(p)e^{-b(p)d^\beta},
        \qquad
        C_{\mathrm{res}}(d)\geq c_{\mathrm{res}}d^\gamma.
    \end{equation}
    If $\PL(d,p)\leq\eps$ for $0<\eps<a(p)$, then
    \begin{equation}
    \begin{aligned}
        d
        &\geq
        \left[
        \frac1{b(p)}
        \log\frac{a(p)}{\eps}
        \right]^{1/\beta},
        \qquad
        C_{\mathrm{res}}(d)
        \geq
        c_{\mathrm{res}}
        \left[
        \frac1{b(p)}
        \log\frac{a(p)}{\eps}
        \right]^{\gamma/\beta}.
    \end{aligned}
    \end{equation}
\end{proposition}

\begin{proof}
    Combining the assumed converse lower bound with the target upper bound gives
    \begin{equation}
        a(p)e^{-b(p)d^\beta}
        \leq
        \PL(d,p)
        \leq
        \eps.
    \end{equation}
    Since $a(p)>0$, division by $a(p)$ yields
    \begin{equation}
        e^{-b(p)d^\beta}
        \leq
        \frac{\eps}{a(p)}.
    \end{equation}
    Taking the logarithm of both positive sides gives
    \begin{equation}
        -b(p)d^\beta
        \leq
        \log\frac{\eps}{a(p)}
        =
        -\log\frac{a(p)}{\eps}.
    \end{equation}
    Multiplying by $-1$ reverses the inequality. Since $b(p)>0$,
    \begin{equation}
        d^\beta
        \geq
        \frac1{b(p)}
        \log\frac{a(p)}{\eps}.
    \end{equation}
    Because $\beta>0$, taking the positive $1/\beta$ power proves
    \begin{equation}
        d
        \geq
        \left[
        \frac1{b(p)}
        \log\frac{a(p)}{\eps}
        \right]^{1/\beta}.
    \end{equation}
    Finally, $\gamma>0$ implies that $d\mapsto d^\gamma$ is increasing. 
    Substituting this distance lower bound into $C_{\mathrm{res}}(d)\geq c_{\mathrm{res}}d^\gamma$ gives
    \begin{equation}
        C_{\mathrm{res}}(d)
        \geq
        c_{\mathrm{res}}
        \left[
        \frac1{b(p)}
        \log\frac{a(p)}{\eps}
        \right]^{\gamma/\beta},
    \end{equation}
    as claimed.
\end{proof}

Consider a depth-$T$ circuit implemented by the logical gadgets of
\cref{def:logical-gadget}. If each gadget has uniform composable diamond-norm error at most $\eps/T$, then
\cref{eq:supp-telescope} gives
\begin{equation}
    \left\|
    \widetilde{\mathcal C}-\mathcal C
    \right\|_\diamond
    \leq
    \sum_{t=1}^{T}\frac{\eps}{T}
    =
    \eps.
\end{equation}
For an exponentially suppressing gadget family, this sufficient per-gadget allocation produces a distance and resource upper bound containing
$\log(T/\eps)$. It does not prove that a
$\log(T/\eps)$ reliability contribution is necessary for every complete
$T$-layer computation. Applying the converse proposition with target
$\eps/T$ would give such a lower bound only if every individual gadget were itself required to have error at most $\eps/T$; the overall error requirement on the complete circuit does not imply this. A necessary
$T$ dependence therefore requires a converse formulated directly for the complete computation.

\subsection{Resource upper bounds from error suppression}

Define $\log_+x=\max\{0,\log x\}$.

\begin{proposition}[Sufficient cost from suppression]\label{prop:generic-upper}
    Fix a physical noise strength $p$ and an integer $M\geq1$. Let
    $\mathcal D\subseteq\mathbb N$ be the set of distances realized by a code-and-gadget family. For every $d\in\mathcal D$ and every
    $m=1,\ldots,M$, consider a logical gadget in the sense of \cref{def:logical-gadget}, implementing one complete ideal logical layer. Such a layer may include logical gates, state preparations, measurements and classical feedforward.
    
    For a permitted preceding measurement record $h$, let
    $\widetilde{\mathcal G}_{d,p,m,h}$ denote the noisy effective logical quantum--classical channel of the $m$th gadget at distance $d$ and physical noise strength $p$, and let $\mathcal G_{m,h}$ denote the corresponding ideal logical-layer channel. Define the uniform composable error by
    \begin{equation}
        q(d,p)
        :=
        \sup_{\substack{1\leq m\leq M\\
        h\ {\rm permitted}}}
        \left\|
        \widetilde{\mathcal G}_{d,p,m,h}
        -
        \mathcal G_{m,h}
        \right\|_\diamond.
        \label{eq:supp-uniform-composable-error}
    \end{equation}
    The diamond norm includes arbitrary joint logical-reference input states and external reference systems, as specified in
    \cref{def:logical-gadget}. Thus ``uniform'' means that the same bound
    $q(d,p)$ applies to every one of the $M$ gadgets and every permitted preceding measurement record. Suppose this error satisfies
    \begin{equation}
        q(d,p)
        \leq
        B(p)e^{-\alpha(p)d^\beta},\qquad
        B(p)>0,\qquad
        \alpha(p)>0,\qquad
        \beta>0,
    \end{equation}
    and suppose its physical circuit size is at most
    \begin{equation}
        C_Hd^\gamma,
        \qquad
        C_H>0,\quad
        \gamma>0.
    \end{equation}
    Here $B(p)$ is the finite-distance error prefactor,
    $\alpha(p)$ is the error-suppression coefficient, $\beta$ is the error-suppression exponent, $C_H$ is the distance-to-circuit-size prefactor, and $\gamma$ is the circuit-size exponent. These quantities are independent of $d$, $M$ and $\eps$, although $B(p)$ and $\alpha(p)$ may depend on the fixed physical noise strength $p$.
    Assume that there are constants $d_0\geq1$ and
    $\kappa\geq1$ such that, for every real target distance $x\geq d_0$, there is an available distance $d\in\mathcal D$ satisfying
    \begin{equation}
        x\leq d\leq\kappa x.
        \label{eq:supp-available-distance}
    \end{equation}
    Thus a target distance is a calculated lower bound sufficient for the desired accuracy, whereas an available distance is the distance of an actual member of the code-and-gadget family.

 {Suppose that these $M$ logical gadgets have total composable error at most}
    \begin{equation}
        C_{\mathrm{comp}}Mq(d,p),
    \end{equation}
    where $C_{\mathrm{comp}}>0$ is independent of $d$, $M$ and $\eps$. For direct diamond-norm telescoping as in
    \cref{eq:supp-telescope}, one may take
    $C_{\mathrm{comp}}=1$.
    Then target error $0<\eps\leq1$ is achieved with per-gadget circuit size
    \begin{equation}
        O\!\left(
        \left(
        1+\frac1{\alpha(p)}
        \log_+\frac{C_{\mathrm{comp}}B(p)M}{\eps}
        \right)^{\gamma/\beta}
        \right).
    \end{equation}
\end{proposition}

\begin{proof}
  {Define the calculated reliability distance}
    \begin{equation}
        {\widehat d:=}
        \left(
        1+\frac1{\alpha(p)}
        \log_+\frac{C_{\mathrm{comp}}B(p)M}{\eps}
        \right)^{1/\beta}.
        \label{eq:supp-target-distance}
    \end{equation}
    Set
    \begin{equation}
        x:=\max\{d_0,\widehat d\}.
    \end{equation}
    By \cref{eq:supp-available-distance}, choose an available
    $d\in\mathcal D$ satisfying
    \begin{equation}
        x\leq d\leq\kappa x.
        \label{eq:supp-chosen-distance}
    \end{equation}
    In particular, \(d\geq\widehat d\).

    Let
    \begin{equation}
        R:=
        \frac{C_{\mathrm{comp}}B(p)M}{\eps}.
    \end{equation}
    By \cref{eq:supp-target-distance},
    \begin{equation}
        \alpha(p)\widehat d^\beta
        =
        \alpha(p)+\log_+R
        \geq
        \log_+R.
    \end{equation}
    Since $d\geq\widehat d$ and $\beta>0$,
    \begin{align}
        C_{\mathrm{comp}}Mq(d,p)
        &\leq
        C_{\mathrm{comp}}B(p)M
        e^{-\alpha(p)d^\beta}
        \leq
        \eps R e^{-\log_+R}
        \leq
        \eps.
    \end{align}
    The last inequality follows from
    $Re^{-\log_+R}\leq1$: it is an equality when $R\geq1$, while for $0<R<1$ it reduces to $R\leq1$.

    It remains to bound the circuit size. Since
    $\widehat d\geq1$, \cref{eq:supp-chosen-distance} gives
    \begin{equation}
        d
        \leq
        \kappa\max\{d_0,\widehat d\}
        \leq
        \kappa\max\{d_0,1\}\widehat d.
    \end{equation}
    Substituting this upper bound for the chosen available distance $d$ into the assumed circuit-size bound gives
    \begin{align}
        C_Hd^\gamma
        &\leq
        C_H
        \left[
        \kappa\max\{d_0,1\}
        \right]^\gamma
        \widehat d^\gamma
        =
        O\!\left(
        \left(
        1+\frac1{\alpha(p)}
        \log_+\frac{C_{\mathrm{comp}}B(p)M}{\eps}
        \right)^{\gamma/\beta}
        \right).
    \end{align}
    The multiplicative factor from replacing the calculated target distance by an available distance is independent of $M$ and $\eps$, so it does not change the asserted asymptotic scaling.
\end{proof}

\subsection{Relation among the results}
{In \cref{tab:theorem-map}, we summarize} the main results and their relation to one another.
\begin{table*}[t]
\caption{Summary of the main results. %
}
\label{tab:theorem-map}
\small
\begin{ruledtabular}
\begin{tabular}{
    p{0.17\textwidth}
    p{0.38\textwidth}
    p{0.37\textwidth}
}
Result & Model and resource & Scientific content \\
\hline
Erasure-noise lower bound
&
Independent erasures; arbitrary ideal adaptive recovery;
worst-case circuit size $\Vpath$ over complete records; variable target error
&
Complete erasure at one time step and a Schmidt-number argument give
$\Vpath=\Omega(S[K+\log(S/\eps)])$.
\\
CSS-code construction
&
$p<\delta_{\mathrm{GV}}$; positive-rate CSS-code existence; ideal recovery
&
One block of width $O(K+\log(S/\eps))$ matches the lower-bound scaling.
\\
General memory lower bound
&
Fixed independent non-unitary qubit channel; arbitrary ideal adaptive control;
sum $\Venv$ of the maximum widths over the individual time steps;
fixed diamond-norm error
&
Every fixed-duration memory satisfies
$\Venv=\Omega(S(K+\log S))$.
Thus no constant relative bound on this quantity holds uniformly for all
widths and durations.
\\
Conditional circuit upper bound
&
Positive-rate blocks; complete logical layers using $O(N)$ qubits and
$O(1)$ time steps; uniform $e^{-\Omega(N)}$ diamond-norm error
&
These sufficient assumptions give physical circuit size
$O(T(K+\log(T/\eps)))$.
\\
Subsystem Singleton bound
&
Exact binary linear subsystem spacetime code and an explicit comparison
between fault coordinates and circuit locations
&
Distance requires additional fault coordinates and, under the comparison
assumption, additional physical circuit size.
\\
Same-syndrome Pauli-fault bound
&
Complete-syndrome Pauli decoder; arbitrary Pauli distribution for the
distribution-overlap form; product noise for the distance form
&
Faults with the same syndrome but different logical actions lower-bound
the logical error and therefore the required distance.
\\
Threshold condition
&
Local-stochastic induced faults and an exponential bound on the number
of malignant fault sets
&
Increasing distance together with a controlled number of malignant sets
is sufficient for a positive threshold; increasing distance alone is not.
\\
\end{tabular}
\end{ruledtabular}
\end{table*}

\section{Width and depth regimes for quantum algorithms}
\label{sec:algorithmic-regimes}

The tight erasure-noise result concerns a quantum memory of logical width
$K$ and storage duration $S$, measured in time steps. Under the assumptions of
\cref{thm:conditional-circuit}, the same additive dependence appears
as an upper bound for a circuit of logical width $K$ and depth $T$. In this
section, {we use $\varepsilon_{\mathrm{FT}}$ for} the desired error of the
fault-tolerant implementation, which is distinct from the approximation
error of the ideal algorithm.

For a fixed-width circuit, {we define} the ideal logical circuit size by
\begin{equation}
G_{\mathrm{log}}=KT,
\end{equation}
where idle logical-qubit locations are included. The conditional circuit result gives 
\begin{equation}
\CFT
=
O\!\left(
T\left(
K+1+\log\frac{AT}{\varepsilon_{\mathrm{FT}}}
\right)
\right).
\label{eq:algorithmic-circuit-scaling}
\end{equation}
{{Here $A$ is the fixed gadget-error prefactor in \cref{ass:gadgets}.} 
Ignoring fixed constants and the lower-order additive term, this has the
form
\begin{equation}
\CFT
=
O\!\left(
G_{\mathrm{log}}
+
T\log\frac{T}{\varepsilon_{\mathrm{FT}}}
\right).
\end{equation}
For the memory problem with independent erasures, setting $T=S$ gives the matching
scaling
\begin{equation}
C_{\min}(K,S,\varepsilon_{\mathrm{FT}})
=
\Theta\!\left(
KS+
S\log\frac{S}{\varepsilon_{\mathrm{FT}}}
\right)
\end{equation}
within the parameter range of \cref{thm:tight-erasure}.

The relative size of the two contributions is determined by the ratio %
\begin{equation}
\begin{aligned}
\frac{\log(T/\varepsilon_{\mathrm{FT}})}{K}&\longrightarrow0
&&\text{when the logarithmic term is asymptotically smaller},\\
\frac{\log(T/\varepsilon_{\mathrm{FT}})}{K}&=\Theta(1)
&&\text{when the two terms have the same order},\\
\frac{\log(T/\varepsilon_{\mathrm{FT}})}{K}&\longrightarrow\infty
&&\text{when the logarithmic term dominates}.
\end{aligned}
\end{equation}
{Accordingly, $K=\Omega(\log(T/\varepsilon_{\mathrm{FT}}))$ is sufficient for constant relative overhead in the conditional circuit upper bound, whereas the logarithmic contribution is negligible only when $K=\omega(\log(T/\varepsilon_{\mathrm{FT}}))$.}
For erasure-noise memory, \cref{eq:main-relative} shows that the minimum relative overhead is $\Theta(1+\log(S/\varepsilon_{\mathrm{FT}})/K)$. For general circuits, \cref{thm:conditional-circuit} gives the corresponding conditional upper bound, but the memory converse does not give a matching lower bound for every algorithm.

\subsection{Polynomial-depth quantum algorithms}

Consider a family of algorithms indexed by an input size $n$. Suppose that
\[
K(n)=\Omega(n),\qquad
T(n)=\operatorname{poly}(n),
\qquad
\varepsilon_{\mathrm{FT}}(n)\geq\frac{1}{\operatorname{poly}(n)}.
\]
Then
\[
\log\frac{T(n)}{\varepsilon_{\mathrm{FT}}(n)}
=
O(\log n),
\]
and therefore
\[
\frac{\log(T/\varepsilon_{\mathrm{FT}})}{K}
=
O\!\left(\frac{\log n}{n}\right)
\longrightarrow 0.
\]
Thus the logarithmic term is asymptotically smaller than the ideal logical circuit size for this class of circuits.

This class includes standard polynomial-depth implementations of quantum
Fourier transforms, reversible arithmetic, variational algorithms of
polynomial depth, polynomial-time Hamiltonian simulation, and Shor's
factoring algorithm. This comparison concerns only the two terms in the conditional upper bound and does not imply that the practical physical overhead of these algorithms is small.

\subsection{Shor's factoring algorithm}

For factoring an $n$-bit integer, standard implementations of Shor's
algorithm use a number of logical qubits that grows at least linearly with
$n$, while their logical depth is polynomial in $n$~\cite{Shor1994}. {Thus}
\[
K(n)=\Omega(n),
\qquad
\log T(n)=O(\log n),
\]
{and hence} $\log(T/\varepsilon_{\mathrm{FT}})/K\to0$ for constant or inverse-polynomial $\varepsilon_{\mathrm{FT}}$.

As a representative numerical example, the abstract circuit estimate of
Gidney and Eker{\aa}~\cite{gidney_how_2021} uses
\[
K
=
3n+0.002n\log_2 n
\]
logical qubits and measurement depth
\[
T
=
500n^2+n^2\log_2 n.
\]
For RSA-2048, these expressions give
\[
K\simeq 6.19\times10^3,
\qquad
T\simeq2.14\times10^9.
\]
Taking $\varepsilon_{\mathrm{FT}}=10^{-3}$ gives
\[
\log\frac{T}{\varepsilon_{\mathrm{FT}}}
\simeq28.4
\]
when the natural logarithm is used, and hence
\[
\frac{\log(T/\varepsilon_{\mathrm{FT}})}{K}\simeq4.6\times10^{-3}.
\]
For this abstract circuit estimate, the logarithmic contribution is therefore small relative to $KT$. This numerical ratio should not be interpreted as a physical
overhead estimate because the asymptotic notation hides noise-dependent and
architecture-dependent constants. In particular, state distillation,
routing, decoding, connectivity and the duration of logical operations can
dominate a practical implementation.

\subsection{Unstructured quantum search}

In the standard oracle model, Grover search over a space of size $N=2^n$ has logical width $K$ and query depth $T$ with the scaling~\cite{Grover1996}
\[
K=\Theta(n),
\qquad
T=\Theta(\sqrt{N})=\Theta(2^{n/2}).
\]
An explicit implementation of the oracle may add further circuit depth.
It follows that
\[
\log T=\Theta(n)=\Theta(K),
\]
so that
\[
\frac{\log(T/\varepsilon_{\mathrm{FT}})}{K}=\Theta(1)
\]
for constant or inverse-polynomial $\varepsilon_{\mathrm{FT}}$. The two contributions in the conditional circuit bound therefore have the same asymptotic order.

The same conclusion applies to amplitude amplification in the query model~\cite{Brassard2002} with
initial success probability $a$. Since its query depth is
\[
T=\Theta(a^{-1/2}),
\]
the relevant ratio is
\[
\frac{\log(T/\varepsilon_{\mathrm{FT}})}{K}
=
\Theta\!\left(
\frac{\log(1/a)+\log(1/\varepsilon_{\mathrm{FT}})}{K}
\right).
\]
Thus the relative size of the two terms depends on the scaling of the initial success probability with the logical width.

\subsection{Phase estimation}

The relative size of the two terms for phase estimation depends on the circuit implementation{~\cite{Kitaev1995,Dobsicek2007}}.

In standard phase estimation, obtaining $m$ bits of phase precision uses an
$m$-qubit control register and controlled evolutions of total duration
approximately $2^m$. If the target system contains $q$ logical qubits, then
\[
K=q+\Theta(m),
\qquad
T=\Theta(2^m),
\]
and therefore
\[
\frac{\log(T/\varepsilon_{\mathrm{FT}})}{K}
=
\Theta\!\left(
\frac{m+\log(1/\varepsilon_{\mathrm{FT}})}{q+m}
\right).
\]
For $q=O(m)$, the logarithmic term and $KT$ may have the same asymptotic order. If $q\gg m$, the logarithmic term is asymptotically smaller.

Iterative phase estimation instead reuses one control qubit. Its width and
depth scale as
\[
K=q+O(1),
\qquad
T=\Theta(2^m)=\Theta(1/\delta),
\]
where $\delta\simeq2^{-m}$ is the desired phase precision. Hence
\[
\frac{\log(T/\varepsilon_{\mathrm{FT}})}{K}
=
\Theta\!\left(
\frac{\log(1/\delta)+
      \log(1/\varepsilon_{\mathrm{FT}})}
     q
\right).
\]
For a fixed-size target system and increasing precision, this ratio diverges and the logarithmic term dominates. This example shows that the comparison depends on the circuit implementation, not only on the abstract computational problem.

\subsection{Amplitude estimation}

A similar distinction occurs in quantum amplitude estimation{~\cite{Brassard2002,Grinko2021}}. Implementations
based on a full phase-estimation register use additional logical qubits to
store the precision, thereby increasing $K$.
Iterative and maximum-likelihood variants reuse a small control register.
For a fixed $q$-qubit problem register and estimation error
$\delta_{\mathrm{alg}}$, their coherent query depth can scale as
\[
T=O(1/\delta_{\mathrm{alg}}),
\]
up to the details of the chosen variant. The corresponding ratio is
\[
\frac{\log(T/\varepsilon_{\mathrm{FT}})}{K}
=
O\!\left(
\frac{
\log(1/\delta_{\mathrm{alg}})
+
\log(1/\varepsilon_{\mathrm{FT}})
}{q}
\right).
\]
For a fixed, small problem register, this ratio can grow with the required precision, in which case the logarithmic term dominates.

Here $\delta_{\mathrm{alg}}$ and $\varepsilon_{\mathrm{FT}}$ have different
meanings. The first is the numerical error of the ideal algorithm, whereas
the second is the probability or channel error introduced by its physical
implementation.

\subsection{Long-time Hamiltonian simulation and quantum signal processing}

Consider a $q$-qubit system simulated under the standard sparse-Hamiltonian or block-encoding access assumptions by a circuit of depth $D$~\cite{BerryChildsKothari2015,LowChuang2019}. The relevant ratio is
\[
\frac{\log(D/\varepsilon_{\mathrm{FT}})}{q}.
\]
If the simulated time and inverse simulation error are polynomial in $q$ {and the required oracle or block encoding has a polynomial-depth implementation},
then $D=\operatorname{poly}(q)$ and this ratio tends to zero.
If
\[
D={\exp(\Theta(q))},
\]
the ratio is $\Theta(1)$. If $q$ is fixed while the simulated time or required polynomial degree tends to infinity, the ratio diverges and the logarithmic term eventually dominates.

This situation can occur in long-time dynamical simulation, high-resolution
spectroscopy, spectral filtering and quantum signal-processing protocols
that repeatedly apply a fixed block encoding to a small logical register.

\subsection{Adiabatic algorithms and quantum walks}

For an adiabatic algorithm on $K=\Theta(n)$ logical qubits with polynomial runtime{~\cite{AlbashLidar2018}},
\[
T=\operatorname{poly}(n),
\]
the ratio $\log(T/\varepsilon_{\mathrm{FT}})/K$ tends to zero. If an exponentially small spectral gap produces
\[
T=\exp(\Theta(n)),
\]
then
\[
\log T=\Theta(n)=\Theta(K),
\]
the ratio is $\Theta(1)$. The logarithmic term dominates only when $\log T$ grows asymptotically faster than the logical width.

The same comparison applies to quantum walks{~\cite{Szegedy2004}}. A walk with polynomial hitting time in $n$ has a vanishing ratio when $K=\Theta(n)$, whereas a walk on $2^n$ states requiring $2^{\Theta(n)}$ coherent steps has a ratio of order one.

\subsection{Long-lived memories, networking and sensing}

The logarithmic term most clearly dominates when a small logical register must remain coherent for a long time. If
\[
K=K_0
\]
is fixed and $T\rightarrow\infty$, then
\[
\frac{\log(T/\varepsilon_{\mathrm{FT}})}{K_0}
\longrightarrow\infty.
\]
Examples include long-lived quantum memories, quantum-network buffers,
memories waiting for heralded entanglement, sequential channel
discrimination and {sensing protocols that accumulate a phase over many
coherent interrogation rounds~\cite{Giovannetti2011}}. 
{When a task requires preservation of an arbitrary quantum state for a prescribed duration, it falls within the memory setting and obeys the corresponding circuit-size lower bound.}

\begin{table*}[t]
\caption{Representative scaling of $\log(T/\varepsilon_{\mathrm{FT}})/K$ for different quantum algorithms.
The entries describe asymptotic circuit families and assume constant or
inverse-polynomial fault-tolerance error unless stated otherwise. The
comparison can change when a different implementation changes the circuit
width or depth.}
\label{tab:algorithmic-regimes}

\small
\begin{ruledtabular}
\begin{tabular}{
p{0.20\textwidth}
p{0.25\textwidth}
p{0.14\textwidth}
p{0.33\textwidth}
}

Workload
&
Representative scaling
&
Ratio $\log(T/\varepsilon_{\mathrm{FT}})/K$
&
Reason
\\
\hline

Polynomial-depth algorithms
&
$K=\Omega(n)$, $T=\operatorname{poly}(n)$
&
$\longrightarrow 0$
&
$\log(T/\varepsilon_{\mathrm{FT}})=O(\log n)\ll K$
\\[0.4em]

Shor factoring
&
$K=\Omega(n)$, $T=\operatorname{poly}(n)$
&
$\longrightarrow 0$
&
The logical register grows much faster than the logarithm of the depth
\\[0.4em]

Grover search on $2^n$ items
&
$K=\Theta(n)$, $T=\Theta(2^{n/2})$
&
$\Theta(1)$
&
$\log T=\Theta(K)$
\\[0.4em]

Standard phase estimation
&
$K=q+\Theta(m)$, $T=\Theta(2^m)$
&
$\Theta(1)$ or $\longrightarrow 0$
&
The precision register contributes $\Theta(m)$ qubits
\\[0.4em]

Iterative phase estimation
&
$K=q+O(1)$, $T=\Theta(1/\delta)$
&
May diverge
&
Diverges when $\log(1/\delta)\gg q$
\\[0.4em]

Iterative amplitude estimation
&
$K=q+O(1)$, $T=O(1/\delta_{\mathrm{alg}})$
&
May diverge
&
A small register is reused for many coherent queries
\\[0.4em]

Long-time simulation or signal processing
&
$K=q$, $T=D$
&
Depends on scaling
&
Tends to zero for $D=\operatorname{poly}(q)$, is $\Theta(1)$ for
$D={\exp(\Theta(q))}$, and diverges when $\log D\gg q$
\\[0.4em]

Adiabatic evolution with exponential runtime
&
$K=\Theta(n)$, $T=\exp(\Theta(n))$
&
$\Theta(1)$
&
The logical width and logarithmic duration have the same order
\\[0.4em]

Long-lived quantum memory
&
$K=O(1)$, $T\rightarrow\infty$
&
$\longrightarrow\infty$
&
The reliability term grows while the stored logical width remains fixed
\\[0.4em]

Quantum-network waiting or sequential sensing
&
$K=O(1)$, $T\rightarrow\infty$
&
$\longrightarrow\infty$
&
A small coherent register is retained across many rounds
\\

\end{tabular}
\end{ruledtabular}
\end{table*}

{\noindent
In \cref{tab:algorithmic-regimes}, we compare the two terms in the conditional circuit upper bound. 
We do not give a matching lower bound for every algorithm in the table. An algorithm may measure and reset ancillas, change its logical width during the computation, or be required to work only for specified inputs, whereas the memory lower bound protects an arbitrary input, including one entangled with a reference. %
We also do not include architecture-dependent costs from decoding, routing, communication, or non-Clifford state preparation.}

\end{document}